\documentclass[12pt, draftclsnofoot, journal, onecolumn]{IEEEtran}
\usepackage{amssymb}
\usepackage{amsmath}
\usepackage{cite}
\usepackage{url}
\usepackage{xcolor}
\usepackage{cite,graphicx,amsmath,amssymb}
\usepackage{subfigure}
\usepackage{citesort}
\usepackage{fancyhdr}
\usepackage{mdwmath}
\usepackage{mdwtab}
\usepackage{caption}
\usepackage{amsthm}
\usepackage{setspace}
\usepackage{booktabs}
\usepackage{algorithm}
\usepackage{algorithmic}
\usepackage{enumerate}
\usepackage{stfloats}

\newtheorem{remark}{Remark}
\newtheorem{theorem}{Theorem}

\newtheorem{lemma}{Lemma}

\newtheorem{corollary}{Corollary}

\newtheorem{definition}{Definition}
\allowdisplaybreaks[4]

\makeatletter
\def\ScaleIfNeeded{%
\ifdim\Gin@nat@width>\linewidth \linewidth \else \Gin@nat@width
\fi } \makeatother

\begin{document}
\title{Cache-enabling UAV Communications: Network Deployment and Resource Allocation}

\author{
Tiankui~Zhang,~\IEEEmembership{Senior Member,~IEEE,}
        Yi~Wang,
        Yuanwei~Liu,~\IEEEmembership{Senior Member,~IEEE,}
        Wenjun~Xu,~\IEEEmembership{Senior Member,~IEEE,}
        and        Arumugam Nallanathan,~\IEEEmembership{Fellow,~IEEE}

\thanks{
Tiankui Zhang, Yi Wang and Wenjun~Xu are with
Beijing University of Posts and Telecommunications, Beijing, China (e-mail: \{zhangtiankui, wangyi96, wjxu\}@bupt.edu.cn).}
\thanks{Yuanwei Liu and Arumugam Nallanathan are with Queen Mary University of London, London, U.K. (e-mail: \{ yuanwei.liu, a.nallanathan\}@qmul.ac.uk).}
}

\maketitle
\begin{abstract}
 In this article, we investigate the content distribution in the hotspot area, whose traffic is offloaded by the combination of the unmanned aerial vehicle (UAV) communication and edge caching. In cache-enabling UAV-assisted cellular networks, the network deployment and resource allocation are vital for quality of experience (QoE) of users with content distribution applications.
 We formulate a joint optimization problem of UAV deployment, caching placement and user association for maximizing QoE of users, which is evaluated by mean opinion score (MOS). To solve this challenging problem, we decompose the optimization problem into three sub-problems. Specifically, we propose a swap matching based UAV deployment algorithm, then obtain the near-optimal caching placement and user association by greedy algorithm and Lagrange dual, respectively. Finally, we propose a low complexity iterative algorithm for the joint UAV deployment, caching placement and user association optimization, which achieves good computational complexity-optimality tradeoff. Simulation results reveal that: i)~the MOS of the proposed algorithm approaches that of the exhaustive search method and converges within several iterations; and ii)~compared with the benchmark algorithms, the proposed algorithm achieves better performance in terms of MOS, content access delay and backhaul traffic offloading.
\end{abstract}

\section{Introduction}
In recent years, unmanned aerial vehicles (UAVs) have been widely used in many industries due to its small size, low price and high flexibility. The characteristics of UAVs make it possible to effectively solve problems in traditional communication, such as high deployment cost and poor adaptability to special scenarios. Therefore, UAV can be deployed as an air base station (BS) to assist the conventional cellular networks~\cite{6477825}. The main application scenarios of UAV communication include high-speed coverage of hotspots, information transmission, emergency communication and so on~\cite{7317490,6978873}. The data traffic requested by mobile users will increase dramatically in future mobile networks. It is predicted that data traffic in the global mobile  networks will reach 2~Zettabyte (ZB) in 2021~\cite{whitepaper}, of which 71\% is used for content distribution. Content caching at the network edge has been proposed as a key enabling technique for content-centric cellular networks to alleviate network traffic load~\cite{6871674}. The key idea of edge caching is to place popular contents close to users such as at the BSs~\cite{8945154} and user terminals~\cite{8868205} so as to reduce the content acquisition delay and backhaul link load. In order to meet the demand of data volume of multimedia content distribution and alleviate the traffic pressure of ground BSs of cellular networks, cache-enabling UAVs are deployed to offload the traffic in the peak hours of some hotspots~\cite{8614433,8576651,8717714,8603721,article,8254370,8626132}, which provides a low cost and rapid deployment solution for content distribution applications with high data rate and low latency requirements. \par

Quality of experience (QoE) is a subjective evaluation of the user's media experience, which has been used as the performance monitor of mobile networks~\cite{5430138}. In order to meet the requirement of high-quality data transmission of video applications, a certain QoE of user needs to be guaranteed.
In this paper, we study the users' QoE maximization in UAV-assisted cellular networks for content distribution. One potential application scenario is that, a stadium that hosts a large-scale sports event, which deploys cache-enabling UAV BSs outside the stadium for hotspots coverage to reduce the traffic load of ground BSs.
\vspace{-1em}
\subsection{Related Works}
Many researchers have carried out research in the field of UAV communication, and some typical problems in UAV communication systems have been discussed, such as UAV deployment, UAV caching deployment, UAV moving trajectory, resource allocation, content transmission security and so on. The relative location between the UAV and users would affect the data rate of content distribution, so the optimized UAV deployment would greatly improve the QoE of users, which has been studied in different scenarios~\cite{8867956,8776639,7486987,9003500,8247211,8432499,9013181}. The optimum placement of a relaying UAV to maximize the capacity of the relay network  was studied in~\cite{8867956}. The UAV's maneuver and power control were jointly optimized to maximize the ground secondary receiver's achievable rate under quasi-stationary UAV scenario and mobile UAV scenario in~\cite{8776639}. In order to maximize the coverage of UAV-assisted cellular networks, the static deployment of multiple UAV BSs in 3D space during UAV flight time was studied in~\cite{7486987}. UAV deployment, channel allocation and relay assignment were jointly optimized, aiming to maximize the capacity of the UAV-aided D2D network in~\cite{9003500}. UAV moving trajectory and communication design were jointly optimized to maximize the minimum throughout in~\cite{8247211}. UAV trajectory and user scheduling were jointly optimized to maximize the minimum worst-case secrecy rate among the users in~\cite{8432499}. Path control of massive UAVs for fast travel and low motion energy without inter-UAV collision was investigated in~\cite{9013181}. Some researches have been conducted in terms of resource allocation, precoding and UAV scheduling to improve network performance in UAV-assisted cellular networks~\cite{7738405,8723343,8648498,8629316,8417640}. A resource allocation optimization mechanism was proposed to minimize mean packet transmission delay in 3D cellular network with multi-layer UAVs in~\cite{7738405}. An energy-efficient resource allocation scheme with the ability of QoE enhancement was proposed in~\cite{8723343}. The joint design of the 3D UAV trajectory and the wireless resource allocation was studied for maximization of the system sum throughput over a given time period in~\cite{8648498}. The sum rate was maximized by jointly optimizing the UAV trajectory and the NOMA precoding in~\cite{8629316}. The proposed UAV scheduling framework was formulated in a generic manner and could be applied in multiple domains comprising short or long-term UAV missions while ensuring uninterrupted service~\cite{8417640}.\par

Edge caching has been a hot research topic in traditional cellular networks~\cite{6871674,8945154,7438747}. The concept of edge caching was proposed in~\cite{6871674}. The content caching and delivery technology of BSs were studied in~\cite{8945154}. In order to improve the QoS and transmission efficiency of the network, optimal caching placement strategy was carried out in~\cite{7438747}. A few research contributions have studied the edge caching combined with UAV communication~\cite{8614433,8576651,8717714,8603721,article,8254370,8626132}. The main purpose of the cache-enabling UAV is to cache popular contents in the UAV BSs related to their associated users so that most frequently requested contents can be served from local caches, instead of forwarding the users' requests over the bandwidth-limited wireless backhaul links. Caching placement was predicted based on content request distribution and optimized by cache space allocation and resource allocation~\cite{8614433}. {The placement of content caching and UAV location are jointly optimized to maximize throughput among IoT devices in \cite{8576651}.} {The content caching and transmission were jointly optimized to maximize the users' reliability in~\cite{8717714}.} In~\cite{8603721}, QoE of users was improved by optimal caching placement strategy. The joint optimization of caching placement and UAV deployment was carried out to maximize QoE in~\cite{article}. A fundamental study on secure transmission of cache-enabling UAV communications was given in~\cite{8626132}. QoE of users, which is as an important indicator for evaluating network performance, has been studied in~\cite{8723343,8603721,article} among above researches.

\vspace{-1em}
\subsection{Motivation and Contribution}
 In this paper, we take maximum QoE as our optimization target of network deployment and resource allocation in cache-enabling UAV-assisted cellular networks. We consider the multiple fixed UAVs deployment scenario, in which, the content access delay and QoE of users are directly related to the relative position between the user and UAV. If the content requested by the user is not cached in the UAV, UAV needs to fetch the content from the ground BS through wireless backhaul link for the associated user. Obviously, caching placement is vital to the content access delay.
 As we have summarized the related work above, there are several works focusing on the UAV deployment and caching placement. However, there are few papers considering the joint optimization of UAV deployment and caching placement. Meanwhile, the network performance of UAV deployment and caching placement is affected by user association which is also an important element that effects the QoE of users given channel bandwidth and transmit power allocation in the networks. Most works in the UAV-assisted cellular networks have ignored the optimization of user association. Motivated by this, we study the joint optimization of UAV deployment, caching placement and user association for QoE maximization. The main contributions of this paper are summarized as follows:

\begin{itemize}
  \item We propose a framework of cache-enabling UAV-assisted cellular networks and take the maximum QoE of users as our optimization target. We use the mean opinion score (MOS) to evaluate the QoE of users. Then, we formulate a joint optimization problem of UAV deployment, caching placement and user association to maximize the QoE of users in the networks.

  \item We propose a joint iterative algorithm to solve the optimization problem. The optimization problem is an integer programming problem which is an NP-hard problem and hard to solve directly. We divide the optimization problem into three sub-problems and solve them by low complexity algorithms respectively. We obtain the UAV deployment by the one-to-one swap matching. Then we obtain the near optimal caching placement and user association by greedy algorithm and Lagrange dual, respectively. Finally, we use the joint iterative algorithm to achieve a {suboptimal} solution. We analyze the computational complexity of the proposed algorithm which has a lower computational complexity than the exhaustive search.

  \item We demonstrate the convergence and network performance to verify the feasibility and effectiveness of the proposed algorithm. We show the convergence of the proposed algorithm by simulation results. The proposed algorithm obtains the {suboptimal} solution with only several iterations, which demonstrates that the computational complexity is greatly reduced at the cost of very small network performance degradation. Meanwhile, compared with the benchmark algorithms, the proposed algorithm achieves better performance in terms of MOS and content access delay of users, as well as the traffic offloading ratio of UAV backhaul links.

\end{itemize}
\subsection{Organization and Notations}
The rest of the paper is organized as follows. The system model and problem formulation are presented in Section II. Section III is the proposed algorithm to solve the optimization problem. In Section IV, we provide numerical simulation results. Finally, conclusions are drawn in Section~V. The main symbols and variables used in this paper are summarized in Table I.
\begin{table}[htbp]
  \begin{center}
      \caption{Main Symbol and Variable List}
      \small
   \begin{tabular}{|c|p{10cm}<{\centering}|}
      \hline
      \textbf{Parameter} & \textbf{Description}  \\       \hline
      $F$ & Number of contents \\ \hline
      $K$ & Number of users \\ \hline
      $M$ & Number of UAVs  \\ \hline
      $N$ & Number of candidate deployment locations  \\ \hline
      $H$ & Cache capacity of UAV   \\ \hline
      $s$ & Size of each content \\ \hline
      $B$ & Downlink bandwidth    \\ \hline
      {$x_{m,n}$} & {Indicator of whether UAV $m$ is deployed in candidate location $n$}\\ \hline
      {$s_{m,f}$} & {Indicator of whether content $f$ cached in UAV $m$}\\ \hline
      {$a_{m,k}$} & {Indicator of whether user $k$ associated with UAV $m$}\\ \hline
      $B_h$ & Backhaul link bandwidth \\ \hline
      $f_c$ & Carrier frequency \\ \hline
      $P_U$ & Transmit power of UAV BS\\ \hline
      $P_M$ & Transmit power of ground BS\\ \hline
      $PL$ & Pathloss  \\ \hline
      $P_{los}, P_{Nlos}$ & Probability of LoS/NLoS link\\ \hline
      ${\mu _{los}},{\mu _{Nlos}}$ & Shadowing random variable\\ \hline
      ${SINR}_{m,n,k}$ & SINR of user $k$ associated with UAV $m$ in candidate location $n$\\ \hline
      ${SINR}_{m,n}$ & SINR of UAV $m$ in candidate location $n$ from MBS\\ \hline
      $r_{m,n,k}$ &  Transmission rate from UAV $m$ in candidate location $n$ to user $k$\\ \hline
      $b_{m,n}$ & Transmission rate from ground BS to UAV $m$ in candidate location $n$\\ \hline
      $D_{m,k}$ & Content access delay of user $k$ associated with UAV $m$  \\ \hline
      ${MOS}_{m,k}$ & MOS of user $k$ associated with UAV $m$ \\ \hline
     \end{tabular}
  \end{center}
  \vspace{-2em}
\end{table}
\section{System Model}
In the cache-enabling UAV-assisted cellular networks, there is one ground macro BS (MBS) with several UAV BSs. Fig.~\ref{UAVs} shows an example of the practical scenario. The ground MBS near the stadium is overloaded which cannot fulfill the traffic requirement of the users in peak hours, for example, the time during a football match. {In this case, the small BSs are overloaded and users can be in terrible communication environment because of the limited frequency resource and SBS capacity. The traffic offloading is assisted by multiple static UAV BSs which are equipped with cache storages.} {The time-frequency resource of the MBS is limited. The congestion will be caused when a large number of users request contents from the MBS at the same time. UAV is used as relay instead of users communicating with the MBS directly.} With limited UAV endurance, UAVs are only deployed to assist ground BSs during peak hours. When the energy is used up, UAV can be recharged or replaced by a new UAV. In this framework, the cache-enabling UAVs save popular multimedia content replicas to serve the users. We define the set of UAVs as $\mathcal{M}= \left\{ {1,2,...,M} \right\}$ and the set of users as $\mathcal{K}= \left\{ {1,2,...,K} \right\}$, respectively. The cache capacity of each UAV is $H$ bits. The UAVs get contents from the ground BS  via wireless backhaul link and proactively cache some popular content replicas in non-peak hours. We assume that the downlink bandwidth of wireless access network is $B$ Hz and the bandwidth of wireless backhaul link is $B_h$ Hz. There is a finite content library, denoted as $\mathcal{F}= \left\{ {1,2,...,F} \right\}$. The size of each content is $s$ bits. A set of $N$ candidate UAV deployment locations, denoted by $\mathcal{N}= \left\{ {1,2,...,N} \right\}$, can be chosen by UAVs for deployment. The location of candidate location $n$ is ${{\bf{w}}_n} = \left( {{x_n},{y_n},{z_n}} \right)$. Each UAV has more than one candidate deployment locations to choose. Let ${x_{m,n}} = 1$ indicate that UAV $m$ is deployed in candidate location $n$, otherwise ${x_{m,n}} = 0$. Then the distance between UAV $m$ and user $k$, UAV $m$ and MBS with $x_{m,n}=1$ are denoted as $d_{m,k} = \sqrt {{{\left\| {{{\bf{w}}_n} - {{\bf{v}}_k}} \right\|}^2}}$, ${d_{m,0}} = \sqrt {{{\left\| {{{\bf{w}}_n} - {{\bf{v}}_0}} \right\|}^2}}$ respectively, where ${\bf{v}}_k$ and ${\bf{v}}_0$ are the location of user $k$ and MBS respectively. Let $q_{k,f} = 1$ indicate that user $k$ requests the content $f$, otherwise $q_{k,f} = 0$. ${a_{m,k}} = 1$ indicates user $k$ is associated with UAV $m$, otherwise ${a_{m,k}} = 0$. One user can only be associated with one UAV, but one UAV can be associated with several users. ${u_{m,f}} = 1$ indicates content $f$ is cached in UAV $m$, otherwise ${u_{m,f}} = 0$. Each UAV can cache $H/s$ contents at most.

\begin{figure} [htbp!]
\centering
\includegraphics[width=0.6\linewidth]{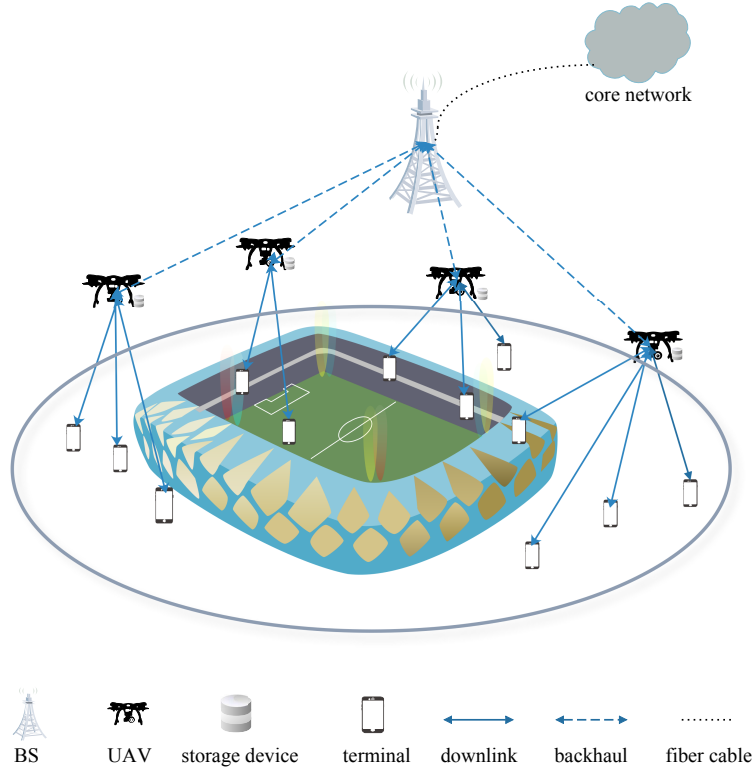}
 \caption{Cache-enabling UAV-assisted cellular networks}
 \label{UAVs}
 \vspace{-2em}
\end{figure}
\subsection{Transmission Model}
The transmission links in this system model follow the UAV channel model provided by 3GPP~\cite{3gpp}. We describe the transmission links between UAVs and users, MBS and UAVs.\par
The propagation channel of the UAV-user and MBS-UAV is modeled by the standard log-normal shadowing model. The standard log-normal shadowing model can be used to model the line-of-sight (LoS) and non-line-of-sight (NLoS) links by choosing specific channel parameters. The received signal-to-interference-plus-noise-ratio (SINR) of user $k$ from UAV $m$ deployed in candidate location $n$ is
\begin{equation}
{S\!I\!N\!R_{m,n,k}} = \frac{{{P_{U}}{{10}^{ - \frac{{P{L_{m,n,k}}}}{{10}}}}}}{{\sum\limits_{m' \ne m, n' \ne n} {{P_{U}}{{10}^{ - \frac{{P{L_{m',n',k}}}}{{10}}}}}  + {\sigma ^2}}},
\end{equation}
where $P_{U}$ is the transmit power of the UAV BS\footnote {Power control is very important to mitigate co-channel interference due to the spectrum sharing among multiple UAVs~\cite{8776639,9018112,8038869}. We mainly focus on the joint optimization of UAV deployment, caching placement and user association in this paper, while leaving power control optimization for our future work.}, and $\sigma ^2$ is the variance of the Gaussian noise. In order to make full use of the spectrum resources, we divide the bandwidth according to the number of users associated with the UAV~\cite{6497017}. So the downlink transmission rate from UAV $m$ deployed in candidate location $n$ to user $k$ is
\begin{equation}
{r_{m,n,k}} = \frac{B}{{\sum\nolimits_{k = 1}^K {{a_{m,k}}}}}{\log _2}\left( {1 + {S\!I\!N\!R_{m,n,k}}} \right).
\end{equation}
The SINR of UAV $m$ deployed in candidate location $n$ from the MBS is
\begin{equation}
S\!I\!N\!R_{m,n} = \frac{P_{M}10^{ - \frac{PL_{m,n}}{10}}}{I_{m,n} + {\sigma ^2}},
\end{equation}
where $P_M$ is the transmit power of the MBS, $I_m$ is the interference from other MBSs to UAV $m$. So the downlink transmission rate from the MBS to UAV $m$ deployed in candidate location $n$ is
\begin{equation}
b_{m,n} = \frac{{{B_{h}}}}{{\sum\nolimits_{k = 1}^K {{a_{m,k}}} }}{\log _2}\left( {1 + S\!I\!N\!R_{m,n}} \right).
\end{equation}
\subsection{Caching Model}
In the cache-enabling UAV-assisted cellular network, UAVs are equipped with cache storage device with limited caching capacity. {If a requested content of a user is cached in its serving UAV, this content would be transmitted to this user via radio downlink directly. Otherwise, the requested content would be first fetched from the core network by its serving UAV via wireless backhaul link with the MBS and then transmitted to this user via radio downlink of its serving UAV.} The content library consists of a limited number of $F$ distinct contents. Each content can be cached in different UAVs, but it can only be placed once in an UAV. We assume the frequency for users to request each of these contents, i.e., popularity of each content, follows a Zipf distribution~\cite{Breslau2002Web}. The popularity distribution of the contents is assumed to remain static over a certain duration~\cite{6175019}. The caching contents of each UAV will be updated regularly. Each user has different request possibility for contents in content library based on the content popularity and requests one content based on content request possibility.
\subsection{Transmission Delay and MOS Model} {In our system model, the transmission delay is divided into two parts, i.e., the downlink radio transmission delay and the backhaul link transmission delay, as shown in Fig.~\ref{UAVs}.}
The downlink radio transmission delay from UAV $m$ to user $k$ is
\begin{equation}
\label{Da}
D_{m,k}^{a} = \frac{{\sum\nolimits_{f=1}^F {s{q_{k,f}}} }}{{{r_{m,n,k}}}}.
\end{equation}
The backhaul link transmission delay from the MBS to UAV $m$ is
\begin{equation}
D_{m,k}^{b} = \frac{{\sum\nolimits_{f=1}^F {\left( {1 - {u_{m,f}}} \right)s{q_{k,f}}} }}{b_{m,n}}.
\end{equation}
If the content requested by user $k$ has been cached in UAV $m$, the backhaul link is no longer needed, that is, $D_{m,k}^{b}=0$ when $u_{m,f}=1$.
The transmission delay from UAV $m$ to user $k$ is
\begin{equation}
{D_{m,k}} = \frac{{\sum\nolimits_{f=1}^F {s{q_{k,f}}} }}{{{r_{m,n,k}}}} + \frac{{\sum\nolimits_{f=1}^F {\left( {1 - {u_{m,f}}} \right)s{q_{k,f}}} }}{{b_{m,n}}}.
\end{equation}
Let ${w_m} = \sum\nolimits_{k=1}^K {{a_{m,k}}}$, we have
\begin{equation}
D_{m,k} = \frac{w_m\sum\nolimits_{f = 1}^F {s{q_{k,f}}}}{{B{{\log}_2}\left( {1 + S\!I\!N\!R _{m,n,k}} \right)}} + \frac{{{w_m}\sum\nolimits_{f = 1}^F {\left( {1 - {u_{m,f}}} \right)s{q_{k,f}}} }}{{{B_h}}{{\log }_2}\left( 1 +S\!I\!N\!R_{m,n} \right)}.
\end{equation}
\par {Inspired by the widely used QoE metric, mean opinion score (MOS) model is used as a measure of the users' QoE for the services like video streaming, content download, or web browsing. As one of the most popular application in wireless networks, we focus on video contents delivery in this paper. The value of MOS is depend on the transmission delay which is an important performance indicator of the mobile networks.} The MOS model is denoted as~\cite{6877621}
\begin{equation}
\label{MOS_model}
MOS_{m,k}= {C_1}\ln \left( {\frac{1}{{{D_{m,k}}}}} \right) + {C_2},
\end{equation}
$C_1$ and $C_2$ are both constants and ${C_1} > 0$. It's obvious that the smaller the delay, the larger the MOS. From the results of our data we set ${C_1}{\rm{ = }}1.120$, ${C_2}{\rm{ = }}4.6746$ so that the value of $MOS_{m,k}$ is ranging from 1 to 5. The higher the score, the better the user's QoE will be.
\subsection{Problem Formulation}
We formulate the joint optimization problem of the UAV deployment, caching placement and user association. The optimization objective is to maximize the MOS of all the users in the cell, which can be expressed as follows
\begin{equation}
\label{op_problem}
\mathop {\max }\limits_{x,a,u} \sum\limits_{m = 1}^M {\sum\limits_{k = 1}^K {{a_{m,k}}MO{S_{m,k}}} },
\end{equation}
According to the definition of MOS in~\eqref{MOS_model}, \eqref{op_problem} can be equivalently expressed as
\begin{equation}
\mathop {\max }\limits_{x,a,u} \sum\limits_{m = 1}^M {\sum\limits_{k = 1}^K {{a_{m,k}}\left( {{C_1}\ln \left( {\frac{1}{{{D_{m,k}}}}} \right) + {C_2}} \right)} }.
\end{equation}
Then we have
\begin{align}
\begin{array}{l}
\;\;\;\sum\limits_{m = 1}^M {\sum\limits_{k = 1}^K {{a_{m,k}}\left( {{C_1}\ln \left( {\frac{1}{{{D_{m,k}}}}} \right) + {C_2}} \right)} } \\
{\rm{ = }}\;{C_1}\sum\limits_{m = 1}^M {\sum\limits_{k = 1}^K {{a_{m,k}}\ln \left( {\frac{1}{{{D_{m,k}}}}} \right)} } {\rm{ + }}\sum\limits_{m = 1}^M {\sum\limits_{k = 1}^K {{a_{m,k}}{C_2}} }
{\rm{ = }}\;{C_1}\sum\limits_{m = 1}^M {\sum\limits_{k = 1}^K {{a_{m,k}}\ln \left( {\frac{1}{{{D_{m,k}}}}} \right)} } {\rm{ + }}K{C_2}.
\end{array}
\end{align}
Let ${Q_{m,k}} = \ln \left( {\frac{1}{{{D_{m,k}}}}} \right)$.
In doing so, the formulated MOS maximization problem is transformed as follows,
\begin{align}
\label{optimal1}
&\mathop {\max }\limits_{x,a,u} \sum\limits_{m = 1}^M {\sum\limits_{k = 1}^K {{a_{m,k}}{Q_{m,k}}}} \\
\label{25a}
&{\rm{s}}{\rm{.t}}{\rm{.}}{\kern 1pt} {\kern 1pt} {\kern 1pt} {\kern 1pt}{a_{m,k}} \in \left\{ {0,1} \right\},\forall m,\forall k,\tag{\ref{optimal1}a}\\
\label{25b}
&{\kern 1pt}{\kern 1pt}{\kern 1pt}{\kern 1pt}{\kern 1pt}{\kern 1pt}{\kern 1pt}{\kern 1pt}{\kern 1pt}{\kern 1pt}{\kern 1pt}{\kern 1pt}{\kern 1pt}{\kern 1pt}{\kern 1pt}{\kern 1pt}{\kern 1pt}{x_{m,n}} \in \left\{ {0,1} \right\},\forall m,\forall n,\tag{\ref{optimal1}b}\\
\label{25c}
&{\kern 1pt}{\kern 1pt}{\kern 1pt}{\kern 1pt}{\kern 1pt}{\kern 1pt}{\kern 1pt}{\kern 1pt}{\kern 1pt}{\kern 1pt}{\kern 1pt}{\kern 1pt}{\kern 1pt}{\kern 1pt}{\kern 1pt}{\kern 1pt}{\kern 1pt}{u_{m,f}} \in \left\{ {0,1} \right\},\forall m,\forall f,\tag{\ref{optimal1}c}\\
\label{25d}
&{\kern 1pt}{\kern 1pt}{\kern 1pt}{\kern 1pt}{\kern 1pt}{\kern 1pt}{\kern 1pt}{\kern 1pt}{\kern 1pt}{\kern 1pt}{\kern 1pt}{\kern 1pt}{\kern 1pt}{\kern 1pt}{\kern 1pt}{\kern 1pt}{\kern 1pt}\sum\limits_{m =1}^M{{a_{m,k}} = 1,\forall k},\tag{\ref{optimal1}d}\\
\label{25e}
&{\kern 1pt}{\kern 1pt}{\kern 1pt}{\kern 1pt}{\kern 1pt}{\kern 1pt}{\kern 1pt}{\kern 1pt}{\kern 1pt}{\kern 1pt}{\kern 1pt}{\kern 1pt}{\kern 1pt}{\kern 1pt}{\kern 1pt}{\kern 1pt}{\kern 1pt}\sum\limits_{f=1}^F s{{u_{m,f}} \le H,} {\kern 1pt} \forall m.\tag{\ref{optimal1}e}
\end{align}
The constrains~\eqref{25a},~\eqref{25b}, and~\eqref{25c} show three binary variables we need to optimize. \eqref{25d} means that one user can only be associated with one UAV. \eqref{25e} is the caching capacity limitation of each UAV.
\section{Proposed Algorithm for Maximum MOS} {Based on the optimization problem formulated in section II, we propose a joint iterative algorithm for maximum MOS in this section.}
The problem we need to solve is the joint optimization of UAV deployment, caching placement and user association. The optimization problem is an NP-hard problem and hard to solve within polynomial time~\cite{booknp}. The variables in this optimization problem are all 0-1 discrete variables. If we use the brute-force search algorithm, the complexity of the algorithm will exceed our computational capacity and it is difficult to achieve. In order to solve this problem effectively, we decompose it into three sub-problems and propose corresponding algorithms:
\begin{enumerate}[\indent(1)]
  \item one-to-one swap matching based UAV deployment algorithm;
  \item greedy based caching placement algorithm;
  \item Lagrange dual based user association algorithm.
\end{enumerate} {The solutions of the three sub-problems are denoted as $\mathcal {X}$, $\mathcal {U}$ and $\mathcal {A}$, respectively.}
\subsection{One-to-One Swap Matching for UAV Deployment}
Although the exhaustive search algorithm can solve the optimization problem of discrete variable, the algorithm complexity is exponential. So it is only suitable for small scale networks. {For practical application, many variations of the basic matching problem have emerged with an array of applications in areas as wide as labor markets, college admissions programs, and communication networks. The one-to-one matching can propose a decentralized algorithm to find a pairwise stable solution with low complexity and fast convergence. So we introduce the one-to-one swap matching to solve UAV deployment sub-problem with fixed caching placement and user association.}

  The set of UAVs $\mathcal{M}$ and the set of candidate locations $\mathcal{N}$ are finite and disjoint sets. When UAV $m$ is deployed in candidate location $n$, the matching pair is denoted as $\left( m ,n \right)$. {In this paper, we build the preference list based on the MOS of users. Each UAV has a preference list over the set of candidate locations. Analogously, each candidate location has a preference list over the set of UAVs. The individual preferences represent the priorities. If UAV $m$ prefers candidate location $n$ to $n'$, we denote it as ${n}{ \succ _{{m}}}{n'}$.} We assume that the preference list of each agent has the following properties:
\begin{itemize}
  \item \textit{complete ordering}: each agent will never confront with an indeterminable choice, i.e., any two alternatives can be compared for an agent to get a preferred one.
  \item \textit{transitive}: it can be denoted as if ${n}{ \succ _{{m}}}{n'}$ and ${n'}{ \succ _{{m}}}{n''}$, then ${n}{ \succ _{{m}}}{n''}$.
\end{itemize}
Based on the above description about one-to-one matching, we give the following definitions~\cite{inproceedings}.
\begin{definition}
A one-to-one matching $\varphi$ is a function from the set $\mathcal{M} \cup \mathcal{N}$ into the set of unordered families of elements of $\mathcal{M} \cup \mathcal{N} \cup \{0\}$, such that
\begin{itemize}
\item $\left| {\varphi \left( m \right)} \right| = 1$ for every $m \in \mathcal{M}$;
\item $\left| {\varphi \left( n \right)} \right| \le 1$ for every $n \in \mathcal{N}$;
\item $n = \varphi \left( m \right) \Leftrightarrow m = \varphi \left( n \right),m \in \mathcal{M},n \in \mathcal{N}$.
\end{itemize}
\end{definition} {The candidate location may not deploy any UAV, but the UAV will be deployed in a certain candidate location.} From the transmission model, we can see that the location of UAV is different, the pathloss between the UAV and the users associated with the UAV will be different, users associated with other UAVs will also be affected. So this is a one-to-one matching with externality. It is not straightforward to define a stability concept since the gains from a matching pair also depends on which the other agents have. Driven by the definition of exchange stable stability, it is convenient to define a swap matching~\cite{8415760,8422358}. Specifically, a swap matching is defined as $\varphi _m^{m'} = \left\{ {\varphi \backslash \left\{ {(m,n),\left( {m',n'} \right)} \right\} \cup \left\{ {\left( {m,n'} \right),\left( {m',n} \right)} \right\}} \right\}$, $\varphi \left( m \right) = n,\varphi \left( m' \right) = n',m,m' \in \mathcal{M},n,n' \in \mathcal{N}$. Based on the swap operation, the definition of a two-sided exchange-stable matching is introduced as follows.
\begin{definition}
A matching $\varphi$ is two-sided exchange-stable if and only if there doesn't exist a pair of agents $\left( m,m'\right)$ with $\varphi \left( m \right) = n,\varphi \left( m' \right) = n'$, such that:
\begin{itemize}
  \item $\forall x \in \left\{ {m,n,m',n'} \right\},{U_x}\left( {\varphi _m^{m'}} \right) \ge {U_x}\left( \varphi  \right)$;
  \item $\exists x \in \left\{ {m,n,m',n'} \right\},{U_x}\left( {\varphi _m^{m'}} \right) > {U_x}\left( \varphi  \right)$.
\end{itemize}
$\left( m,m'\right)$ is called a blocking pair.
\end{definition}
${U_x}\left( \varphi  \right)$ is the utility of agent $x$ under matching $\varphi$. The characteristics of the blocking pair ensure that if a swap matching is approved, the achievable utility of any agent involved will not decrease and at least one agent's utility will increase. The definition indicates that a swap matching is two-sided exchange-stable when there doesn't exist any blocking pair in $\mathcal{M} \cup \mathcal{N} \cup \{0\}$. To avoid the meaningless cycle of swap matching, we ensure the number of swap between UAVs and candidate locations is less than 2.\par

As discussed above, the UAV deployment problem is a one-to-one matching problem with externality. To model the externality, the preference list of UAV $m$ for candidate location $n$ is formulated as the sum of user's MOS associated with the UAV, which is denoted as follows,
\begin{equation}\label{pre1}
\mathcal{U}_m^n= \sum\limits_{k=1}^K {{a_{m,k}}MOS_{m,k}},\;\;\;\varphi \left( m \right) = n.
\end{equation}
The preference list of the candidate location $n$ for UAV $m$ is denoted as
\begin{equation}\label{pre2}
\mathcal{U}_n^m = \left\{ \begin{array}{l}
\sum\limits_{k=1}^K {{a_{m,k}}MOS_{m,k},\;\;\;\varphi \left( n \right) = m}, \\
0 ,\quad \quad \quad \quad \quad \;\;\;\;\;\;\varphi \left( n \right) = 0,
\end{array} \right.
\end{equation}
where $\varphi \left( n \right) = 0$ represents that there is no UAV deployed in candidate location $n$.\par {Specifically, for UAV $m$, any two candidate locations $n$ and $n'$, and any two matchings $\varphi$ and $\varphi^{'}$, we have the following relations,}
\begin{equation}
\label{prefer}
\left( {n,\varphi } \right){ \succ _{m}}\left( n',\varphi ' \right) \Leftrightarrow {\mathcal{U}^n}\left( \varphi  \right) > {\mathcal{U}^{n'}}\left( {\varphi '} \right),
\end{equation}
which implies that UAV $m$ prefers the candidate $n$ to the $n'$ only if UAV $m$ can achieve a higher MOS in $n$ than $n'$.\par {We utilize the Gale-Shapley (GS) algorithm proposed in~\cite{stable} to construct the initial matching state between UAVs and candidate locations. In the GS based initialization procedure, we calculate the SNR between UAV and users instead of SINR. Then the UAVs and candidate locations can build their own preference lists by~\eqref{pre1} and~\eqref{pre2}. Based on the established preference lists, each candidate location proposes to the favorite UAV based on its preference list. At the UAV acceptance phase, each UAV accepts the candidate location with prior preference and rejects others. The algorithm terminates when all UAVs have been matched to the candidate locations or every unmatched location has been rejected by every UAV. Based on the initial matching state, the swap operation procedure is employed to further enhance the utility. The process of one-to-one swap matching based UAV deployment algorithm is summarized in \textbf{Algorithm~\ref{MLD}}.}
\begin{algorithm}
\caption{Swap matching based UAV deployment algorithm}
\label{MLD}
\begin{algorithmic}[1]
\STATE  Construct the initial UAV-DC matching state $\mathcal{S}_I$ based on the GS algorithm. The matching state is denoted as $\mathcal{S}$. Let $\mathcal{S}=\mathcal{S}_I$
\REPEAT
\STATE  For any UAV $m \in\mathcal{S}_I$, it searches for another UAV $m' \in \mathcal{S}_I\backslash \mathcal {S}_I$$\left( {\varphi \left( m \right)} \right)$
\IF{$m' \ne 0$ }
\IF{$\left( {m,m'} \right)$ is a blocking pair}  
\STATE Swap $\left( {m,m'} \right)$, $\varphi  = \varphi _m^{m'}$
\ELSE
\STATE  Keep the current matching state
\ENDIF
 \ELSE
\IF{${U_m}\left( {\varphi _m^0} \right) \ge {U_m}\left( \varphi  \right)$}
\STATE Swap $\left( {m,0} \right)$, $\varphi  = \varphi _m^0$
\ELSE
\STATE Keep the matching state
\ENDIF
\ENDIF
\UNTIL{No blocking pair in the matching}
\STATE  \textbf{Output:}  matching state $\varphi$\
\end{algorithmic}
\end{algorithm}

The complexity, convergence and stability of \textbf{Algorithm~\ref{MLD}} are analyzed as following.\par
\emph{(1) Complexity:} {The GS algorithm requires each candidate location to propose to one UAV based on its preference list, and each UAV accepts its favorite candidate location.} The computational complexity of the initialization GS algorithm is $\mathcal{O}\left( {MN} \right)$. In the swap matching process, there are at most $N-1$ candidate locations for each candidate location in each iteration to swap. For a given number of total iteration $L$, the complexity is $\mathcal{O}\left( {C_N^2{ML}} \right)$. Hence, the complexity of \textbf{Algorithm~\ref{MLD}} is $\mathcal{O}\left( {MN{\rm{ + }}C_N^2{ML}} \right)$.\par
\emph{(2) Convergence and Stability:} According to \textbf{Algorithm~\ref{MLD}}, any UAV cannot find another candidate location to form a swap-blocking pair under the current matching $\varphi$. Hence, a two-sided exchange-stable matching is formed between UAVs and candidate locations. Since the utility function will increase monotonically by the swap operation in \textbf{Algorithm~\ref{MLD}} and the utility function is bounded due to the bandwidth constraint, \textbf{Algorithm~\ref{MLD}} would reach a local solution after finite swap operations and converge to a two-sided exchange-stable status. However, not all two-sided exchange-stable matching are local optimal.

\subsection{Greedy Algorithm for Caching Placement}
Next we need to solve the caching placement. For each UAV in the network, caching strategy is independent of each other, when user association and UAV deployment are both determined. Caching placement is directly related to the preference of the users associated with the UAV. {In this case, the optimization sub-problem of  caching placement is denoted as}
\begin{align}\label{caching placement}
\mathop {\max }\limits_u&\sum\limits_{k = 1}^K {{a_{m,k}}\ln \left( {\frac{1}{{{D_{m,k}}}}} \right)}\\
&{\rm{s}}{\rm{.t}}{\rm{.}}~\eqref{25c},~\eqref{25e}\notag,
\end{align}
where \eqref{25c} denotes that $u_{k,f}$ is a 0-1 binary variable.~\eqref{25e} denotes that the cache space of UAV is limited.

\begin{definition}
Let $\mathit{g}:2^{\mathcal{G}} \to \mathbb{R}$ represents a set function. When the following two conditions are satisfied, we say that $\mathit{g}$ is a monotonic and submodular function set.
\begin{enumerate}[\indent(1)]
  \item For every ${\cal X} \subseteq {\cal Y} \subseteq {\cal G}$, we have $g\left( {\cal X} \right) \le g\left( {\cal Y} \right)$;
  \item For every ${\cal X} \subseteq {\cal Y} \subseteq {\cal G}$, and $x \in {\cal G}\backslash {\cal Y}$ we have:\par
   $g\left( {{\cal X} \cup \left\{ x \right\}} \right) - g\left( {\cal X} \right) \ge g\left( {{\cal Y} \cup \left\{ x \right\}} \right) - g\left( {\cal Y} \right)$, \\where $\mathcal{G}$ is the ground set.
\end{enumerate}
\end{definition}
\par
\begin{theorem}
\label{lemma_cp}
$Q\left( u\right )$ is a monotone and submodular function. We can obtain a near-optimal solution by greedy algorithm within polynomial time.
\end{theorem}
\begin{proof}
See Appendix A for a detailed proof.
\end{proof}
We can obtain a near-optimal solution by \textbf{Algorithm~\ref{cache_placement}} within polynomial time.
\begin{algorithm}[htbp!]
\caption{Greedy based caching placement algorithm}
\label{cache_placement}
\begin{algorithmic}[1]
\STATE \textbf{Input:} the feasible solution set $\cal I$, the current caching placement $\mathcal{U}$ and the content library $\mathcal{F}$
\STATE \textbf{Initialize:} left cache space ${C_m} \leftarrow H$, the current caching placement $\mathcal{U} \leftarrow \emptyset$, $CF \leftarrow \mathcal{F}$
\WHILE {$C_m > 0$}
\STATE Choose content $f^{*}$ for UAV $m$ by
\STATE ${f^*} = \mathop {\arg \max }\limits_{f \in CF} \left( \sum\limits_{k = 1}^K {Q\left( {\mathcal{U} \cup {u_{m,f}}} \right) - \sum\limits_{k = 1}^K Q\left( {\mathcal{U}} \right)} \right)$
\IF {$\left( {{u_{m,f^{*}}} \cup \mathcal{U}} \right) \subseteq \mathcal{I}$}
\STATE $u_{m,f^{*}}=1$
\STATE $\mathcal{U} \leftarrow \mathcal{U} \cup {u_{m,f^{*}}}$
\STATE $C_m\leftarrow C_m-s$
\STATE $CF \leftarrow CF/{f^*}$
\ELSE
\STATE Keep the current caching placement
\ENDIF
\ENDWHILE
\STATE \textbf{Output:} caching placement $\mathcal{U}$
\end{algorithmic}
\end{algorithm}
\par
In \textbf{Algorithm~\ref{cache_placement}}, the compare between the gains is based on merge-sort algorithm. In the worst case, the complexity of \textbf{Algorithm~\ref{cache_placement}} is $\mathcal {O} \left( {M\left( {F\left( {\ln F + 1} \right) + I} \right)} \right)$, where $I = H/s$.
\begin{remark}\label{greed_remark}
In \textbf{Algorithm~\ref{cache_placement}}, the decision whether a content can be cached in the UAV is decided by the gain the content can bring. In each iteration, we cache a content in the UAV that maximizes $Q$. The greedy algorithm is guaranteed to find a solution within a $\left( {1 - {1 \mathord{\left/
 {\vphantom {1 e}} \right.\kern-\nulldelimiterspace} e}} \right)$ factor of the optimal
solution. Although we do not have accuracy guarantees of the performance gap between the solution obtained from \textbf{Algorithm 2} and the optimal solution, we show that this performance gap on MOS is quite small in Section IV.
\end{remark}
\subsection{Lagrange Dual Algorithm for User Association}
After solving the UAV deployment and caching placement, we can take $\mathcal {X}$ and $\mathcal {U}$ as fixed matrix. The user association problem is an optimization problem with 0-1 binary variable under constrains, which is denoted as,
\begin{align}\label{optimal_user_association}
\mathop {\max }\limits_a &\sum\limits_{m = 1}^M {\sum\limits_{k = 1}^K {{a_{m,k}}\ln \left( {\frac{1}{{{D_{m,k}}}}} \right)}} \\
\label{22a}
&{\rm{s}}{\rm{.t}}{\rm{.}}\sum\limits_{k=1}^K {{a_{m,k}} = {w_m}},\tag{\ref{optimal_user_association}a}\\
&~~~~~\eqref{25a},~\eqref{25d}.\notag
\end{align}
Let ${T_{m,k}} = {w_m \mathord{\left/{\vphantom {1 {{D_{m,k}}}}} \right.\kern-\nulldelimiterspace} {{D_{m,k}}}}$. Then:
\begin{equation}\label{t_mk}
{T_{m,k}} = \frac{1}{{\frac{{\sum\nolimits_{f = 1}^F {s{q_{k,f}}} }}{{B{{\log }_2}\left( {1 + S\!I\!N\!R_{m,n,k}} \right)}} + \frac{{\sum\nolimits_{f = 1}^F {\left( {1 - {u_{m,f}}} \right)s{q_{k,f}}} }}{{{B_{h}}{{\log }_2}\left( {1 + S\!I\!N\!R_{m,n}} \right)}}}}.
\end{equation}
\par
When the network scale is small, the optimal user association can be found through a brute force search. The complexity of the brute force algorithm is $\mathcal{O}\left( {M^K} \right)$, where $M$ and $K$ are the number of UAVs and users, respectively. The computation is essentially impossible for a modest-sized network. So we propose a low complexity algorithm to solve this problem. The only coupling constraint is $\sum\nolimits_k {{a_{m,k}} = {w_m}}$. This motivates us to turn to the Lagrange dual decomposition method whereby a Lagrange multiplier $\alpha$ is introduced to relax the coupling constraint~\cite{6497017}:
\begin{equation}
L\left( \alpha  \right){\rm{ = }}{f_a}\left( \alpha  \right) + {g_w}\left( \alpha  \right).
\end{equation}
The dual problem is
\begin{align}
\label{Lagrange}
\mathop {\min }\limits_\alpha  L&\left( \alpha  \right) = {f_a}\left( \alpha  \right) + {g_w}\left( \alpha  \right)\\
&{\rm{s}}{\rm{.t}}{\rm{.}}~\eqref{25a},~\eqref{25d},\notag
\end{align}
where
\begin{equation}
\label{aaa}
f\left( \alpha  \right) = \mathop {\max }\limits_a  \sum\limits_{m = 1}^M {\sum\limits_{k = 1}^K {{a_{m,k}}\left( {\ln \left( {{T_{m,k}}} \right) - {\alpha _m}} \right)}},
\end{equation}
\begin{equation}
g\left( \alpha  \right) = \mathop {\max }\limits_w  \sum\limits_{m=1}^M {{w_m}\left( {{\alpha _m} - \ln \left( {{w_m}} \right)} \right)}.
\end{equation}
\begin{equation}
\begin{aligned}
L\left( \alpha  \right){\rm{ = }}&\sum\limits_{m = 1}^M {\sum\limits_{k = 1}^K {{a_{m,k}}\ln \left( {{T_{m,k}}} \right)} }  - \sum\limits_{m = 1}^M {{w_m}\ln } \left( {{w_m}} \right) - \sum\limits_{m = 1}^M {{\alpha _m}\left( {\sum\limits_{k = 1}^K {{a_{m,k}}}  - {w_m}} \right)}.
\end{aligned}
\end{equation}
Given the dual variable $\alpha_m$, the solution of maximizing the Lagrangian with respect to $a_{m,k}$ can be explicitly obtained by
\begin{equation}\label{choose UAV}
{a_{m,k}} = \left\{ \begin{array}{l}1\quad \quad if\quad m = {m^ * },\\0\quad \quad otherwise,
\end{array} \right.
\end{equation}
where
\begin{equation}
{m^{*}} = \mathop {\arg \max }\limits_m \left( {\ln \left( {{T_{m,k}}} \right) - {\alpha _m}} \right).
\end{equation}
Taking the second-order derivative of the Lagrangian w.r.t. $w_m$ yields
\begin{equation}
\frac{{\partial {L^2}}}{{{\partial ^2}{w_m}}} =  - \frac{1}{{{w_m}}} < 0.
\end{equation}
This means that the Lagrangian is a concave function of $w_m$.
\begin{equation}
\frac{{\partial L}}{{\partial {w_m}}} = {\alpha _m} - \ln \left( {{w_m}} \right) - 1.
\end{equation}
By setting $\frac{{\partial L}}{{\partial {w_m}}} = 0$, the optimal value of $w_m$ is given by
\begin{equation}
\label{wm}
{w_m}^ *  = {e^{{\alpha _m} - 1}}.
\end{equation}
The value of the Lagrange multiplier $\alpha$ is updated by
\begin{equation}
\label{updatem}
{\alpha _m}\left( {t + 1} \right) = {\left[ {{\alpha _m}\left( t \right) - \delta \left( t \right)\left( {{w_m}\left( t \right) - \sum\limits_{k = 1}^K {{a_{m,k}}} } \right)} \right]^ + },
\end{equation}
where ${\left[ a \right]^ + } = \max \left\{ {a,0} \right\}$, $t$ is the iteration index, and $\delta \left( t \right)$ is dynamically chosen step size sequence based on some suitable estimates.\par {We propose the Lagrange dual algorithm to obtain the near-optimal user association as given in \textbf{Algorithm 3}~\cite{6497017}.}
\begin{algorithm}
\caption{Lagrange dual based user association algorithm}
\label{user_association}
\begin{algorithmic}[1]
\STATE \textbf{Initialize:} every user calculates $T_{m,k}$ according to~\eqref{t_mk}; $t=0$ and ${\alpha _m}\left( 0 \right),\forall m$. The initial user association is $\mathcal{A}_I$. The user association is denoted as $\mathcal{A}$. Let $\mathcal{A} = \mathcal{A}_I$
\REPEAT
\STATE Users choose the serving UAV according to~\eqref{choose UAV}
\STATE Update the user association $\mathcal{A}$
\STATE Calculate the corresponding ${w_m}^{*}\left( t \right)$ by~\eqref{wm}
for each UAV
\STATE Update ${\alpha _m}\left( {t + 1} \right)$ by~\eqref{updatem}
\STATE $t \leftarrow t + 1$
\UNTIL{Convergence}
\STATE \textbf{Output:} user association $\mathcal{A}$
\end{algorithmic}
 \vspace{-0.2em}
\end{algorithm}

The multiplier $\alpha$ works as a message between UAVs and users in the network. In fact, it can be interpreted as the price of the UAVs determined by the load situation, which can be either positive or negative. If we interpret $\sum\limits_k {{a_{m,k}}}$ as the serving demand for UAV $m$ and $w_m$ as the service UAV $m$ can provide. Then $\alpha_m$ is the bridge between demand and supply. From~\eqref{updatem}, if the demand $\sum\limits_k {{a_{m,k}}}$ is larger than the supply $w_m$, the price $\alpha_m$ will increase. On the contrary, the price $\alpha_m$ will decrease. Thus, when UAV $m$ is overloading, $\alpha_m$ will increase and fewer users will associate with it, while the price of other under-loaded UAVs will decrease so as to attract more users. \par
\emph{(1) Complexity:} At each iteration, the complexity of the distributed algorithm is $\mathcal{O} \left(MK\right)$. For a given number of total iteration $L^{'}$, the algorithm is guaranteed to converge to a near-optimal solution.\par
\emph{(2) Step Size and Convergence:} The step size is non-summable and diminishing~\cite{article1}. We assume that step size updates according to the following rule
\begin{equation}
\label{updatedelta}
\begin{aligned}
\delta \left( t \right) &= \lambda \left( t \right)\frac{{L\left( {\alpha \left( t \right)} \right) - L\left( t \right)}}{{{{\left\| {\partial L\left( {\alpha \left( t \right)} \right)} \right\|}^2}}}\\
&0 < \lambda  \le \lambda \left( t \right) \le \mathop \lambda \limits^ - < 2,
\end{aligned}
\end{equation}
where $\lambda$ and $\mathop \lambda \limits^ -$ are both scalars. $L\left( t \right)$ is an estimate of the optimal value ${L^*}$ of the optimization problem. $L\left( t \right)$ updates according to the following rule
\begin{equation}
\label{updateL}
L\left( t \right) = \mathop {\min }\limits_{0 \le \tau  \le t} L\left( {\alpha \left( t \right)} \right) - \varepsilon \left( t \right),
\end{equation}
$\varepsilon \left( t \right)$ updates according to the following rule
\begin{equation}
\label{updateepsilon}
\varepsilon \left( {t + 1} \right) = \left\{ \begin{array}{l}
\rho \varepsilon \left( t \right)\quad  if\;L\left( {\alpha \left( {t + 1} \right)} \right) \le L\left( {\alpha \left( t \right)} \right),\\
\max \left\{ {\beta \varepsilon \left( t \right),\varepsilon } \right\}\quad \,otherwise,
\end{array} \right.\;
\end{equation}
$\varepsilon$, $\beta$ and $\rho$ are all positive constants with $\beta  < 1$, $\rho >1$, respectively.
The target level of $L\left( t \right)$ can be obtained by appropriate $\varepsilon \left( t \right)$. Whenever the target level is achieved, we increase $\varepsilon \left( t \right)$ or keep it at the same value. If the target level is not attained, $\varepsilon \left( t \right)$ is dropped to the threshold $\varepsilon$. Hence, we have the following theorem.
\begin{theorem}
\label{theorem}
Assume that the step size $\delta \left( t \right)$ is updated by~\eqref{updatedelta} with the adjustment procedure~\eqref{updateL} and~\eqref{updateepsilon}. If the optimal value ${L^{\rm{*}}} >  - \infty $ then
\begin{equation}
\mathop {\inf }\limits_t L\left( t \right) \le {L^*} + \varepsilon.
\end{equation}
\end{theorem}
\begin{proof}
See Appendix B for a detailed proof.
\end{proof}

\subsection {Suboptimal Solution for Optimization Problem}
Since we solve the optimization problem by decomposing the problem into three sub-problems, we can obtain the {suboptimal} solution by alternate iteration based on the algorithms proposed above. The proposed joint UAV deployment, user association and caching placement algorithm is described as follows.
\begin{algorithm}
\caption{Joint UAV deployment, user association and caching placement algorithm}
\label{maxMOS}
\begin{algorithmic}[1]
\STATE \textbf{Initialize:} UAV deployment $\mathcal{X}$, user association $\mathcal{A}$ and caching placement $\mathcal{U}$, $l=1$, $MOS\left( 0 \right) = 0$
\REPEAT
\STATE update $\mathcal{X}$ by Algorithm \ref{MLD}
\STATE update $\mathcal{U}$ by Algorithm \ref{cache_placement}
\STATE update $\mathcal{A}$ by Algorithm \ref{user_association}
\STATE calculate $MOS\left( l \right)$ by~\eqref{op_problem}
\STATE $l \leftarrow l+1$
\UNTIL{$\left| {MOS\left( l \right) - MOS\left( {l - 1} \right)} \right| < \delta$}
\STATE \textbf{End}
\end{algorithmic}
\end{algorithm} {
\begin{theorem}\label{theorem3}
Algorithm \ref{cache_placement} can yield an increasing objective value in each iteration until convergence.
\end{theorem}
\begin{proof}
See Appendix C for a detailed proof.
\end{proof}
}
In the following, we discuss the complexity and the convergence of \textbf{Algorithm~\ref{maxMOS}}.\par {\emph{(1) Complexity:}
During each iteration, three subproblem algorithms are performed to solve three subproblems. The complexity of \textbf{Algorithm~\ref{MLD}},~\textbf{\ref{cache_placement}},~\textbf{\ref{user_association}} has been analyzed above in Subsection A, B, C. We assume that the proposed algorithm can obtain the suboptimal solution with G iterations. So the complexity of \textbf{Algorithm~\ref{maxMOS}} is $\mathcal{O}\left(G  \left( {MN^2L + M\left( {F\left( {\ln F + 1} \right) + I} \right)+ MKL^{'}}\right)\right)$, where $L$ and $L^{'}$ are the iteration number given in \textbf{Algorithm~\ref{MLD}} and \textbf{Algorithm~\ref{user_association}}, respectively.}\par {\emph{(2) Convergence:} The proposed algorithm can reach convergence after several iterations.
 \begin{proof}
 See Appendix D for a detailed proof.
 \end{proof}}
\begin{remark}\label{tradeoff}
For the optimization problem, we decompose it into three sub-problems and propose three low complexity algorithms, respectively. The proposed algorithm makes a tradeoff between the network performance and the computational complexity. The complexity of the proposed algorithm is greatly reduced at the cost of very small network performance degradation.
\end{remark}
\section{Numerical Results}
In this section, the effectiveness of the proposed algorithm is verified by comparing with two benchmark algorithms. {In the simulation, K users are randomly distributed in a cell and a MBS is deployed 1km away from the cell.} The coverage radius of each UAV is $100$~m. The height of UAV is ranging from 45~m to 60~m. {The simulation area is divided into $N$ areas, and the candidate positions in each sub-area are uniformly distributed.} In the practical application, the candidate locations can be determined in advance according to the coverage region and the obstacles of a certain hotspot. {The popularity of $F$ contents follows a \emph{Zipf}-like distribution. Without loss of generality, we rank these contents in a descending order according to their popularities. The popularity of the $i^{\rm{th}}$ content is denoted as
\begin{equation}
{\rho _i} = \frac{{{1 \mathord{\left/
 {\vphantom {1 {{i^\gamma }}}} \right.
 \kern-\nulldelimiterspace} {{i^\gamma }}}}}{{\sum\nolimits_{f = 1}^F {{1 \mathord{\left/
 {\vphantom {1 {{f^\gamma }}}} \right.
 \kern-\nulldelimiterspace} {{f^\gamma }}}} }},
\end{equation}
where the Zipf parameter $\gamma$ determines the skewness in the users' preference. The pathloss of LoS and NLoS link is denoted as
\begin{equation}
\label{pathloss}
\small{
{PL} = \left\{ \begin{array}{l}
30.9 + \left( {22.25 - 0.5{{\log}_{10}}{h}} \right){\log _{10}}{d} + 20{\log _{10}}{f_c}+{\mu _{los}},\rm{for\,~LoS\,~link,}\\
\max \left\{ {PL^{LOS},32.4 + \left( {43.2 - 7.6{{\log }_{10}}{h}} \right){{\log }_{10}}{d} + 20{{\log }_{10}}{f_c}} \right\}+{\mu _{Nlos}},\rm{for\,~NLoS\,~link.}
\end{array} \right.}
\normalsize
\end{equation}
The LoS/NLoS link is stochastically determined by the LoS possibility $P_{LOS}$ which is denoted as
\begin{equation}
\label{possibility}
{P_{LOS}} = \left\{ \begin{array}{l}
1,\;if\,\sqrt {{d^2} - {h^2}}  \le {d_0},\\
\frac{{{d_0}}}{{\sqrt {{d^2} - {h^2}} }} + \exp \left\{ {\left( {\frac{{ - \sqrt {{d^2} - {h^2}} }}{{{p_1}}}} \right)\left( {1 - \frac{{{d_0}}}{{\sqrt {{d^2} - {h^2}} }}} \right)} \right\},\;if\,\sqrt {{d^2} - {h^2}}  > {d_0},
\end{array} \right.
\end{equation}
where $d$ is the distance between UAV and user/MBS which is related with the UAV deployment location, $f_c$ is the carrier frequency, $\mu_{LoS}$ and $\mu _{Nlos}$ are the shadowing random variable for LoS link and NLoS link, respectively. ${d_0} = \max \left[ {295.05{{\log }_{10}}h - 432.94,18} \right]$, ${\mu _{los}} = 4.64{e^{ - 0.0066h}}$, ${\mu _{Nlos}} = 6$ and ${p_1} = 233.98{\log _{10}}h - 0.95$. The NLoS probability is ${P_{NLoS}} = 1-{P_{LOS}}$. This model holds for the given altitude $22.5~\rm{m} \le h \le 300~\rm{m}$. The parameter setting for UAV-user and MBS-UAV communication is based on 3GPP~\cite{3gpp}. It is obvious that UAV altitude can influence the value of pathloss. We will investigate the influence of UAV height by simulation to further optimize the system design.} The detailed simulation parameters are given in Table II.
\begin{table}[!htbp]
  \begin{center}
      \caption{Simulation Parameters}
    \begin{tabular}{|p{7cm}<{\centering}|c|}
      \hline
      \textbf{Parameter} & \textbf{Value}  \\       \hline
      Number of UAVs & 4\\ \hline
      Number of candidate locations of UAVs & 12\\ \hline
      Number of contents & 200\\ \hline
      Size of each content & 10 Mbits\\ \hline
      Zipf parameter $\gamma$ & 0.6, 1\\ \hline
      Convergence gap & $10^{-3}$\\ \hline
      Bandwidth & 20 MHz\\ \hline
      Carrier frequency & 2 GHz\\ \hline
      Transmit power of UAV & 23 dBm\\ \hline
      Transmit power of MBS & 46 dBm\\ \hline
      Variance of the Gaussian noise $\sigma ^2$ & -174~dBm/Hz\\ \hline
     \end{tabular}
  \end{center}
\end{table}

 We compare the proposed algorithm with two benchmark algorithms, classic algorithm and  random algorithm. In the classic algorithm, the UAV deployment is subject to uniform distribution, caching placement is decided by max-popular caching placement, and user association is decided by max-C/I access. In the random algorithm, the UAV deployment, caching placement and user association are all subject to the random distribution. We demonstrate the effectiveness of the proposed algorithm on the network performance, including the average MOS, UAV backhaul traffic offloading ratio. The average MOS of users is denoted as
\begin{equation}
MOS_{ave} = \frac{1}{K}\sum\limits_{m = 1}^M {\sum\limits_{k = 1}^K {{a_{m,k}}MO{S_{m,k}}}}.
\end{equation} {In the simulation process, we first calculate $\ln \left( {{1 \mathord{\left/
 {\vphantom {1 {{D_{m,k}}}}} \right.
 \kern-\nulldelimiterspace} {{D_{m,k}}}}} \right)$. Prober $C1$ and $C2$ are set  according to $\ln \left( {{1 \mathord{\left/
 {\vphantom {1 {{D_{m,k}}}}} \right.
 \kern-\nulldelimiterspace} {{D_{m,k}}}}} \right)$ so that the MOS value is between 1 and 5.}\par
The UAV backhaul traffic offloading ratio ratio is denoted as
\begin{equation}
\label{offloading}
O = \frac{1}{K}\sum\limits_{m = 1}^M {\sum\limits_{k = 1}^K {{a_{m,k}}{q_{k,f}}{u_{m,f}}}}.
\end{equation}
\begin{figure} [htbp!]
\centering
\includegraphics[width=0.6\linewidth]{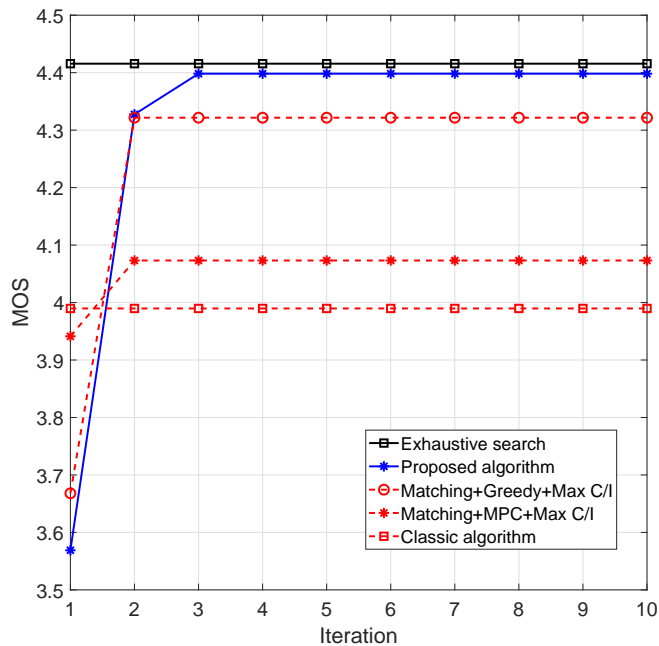}
 \caption{Convergence of the proposed algorithm.}\label{conver}
 \vspace{-2em}
\end{figure}
First, we demonstrate the convergence of the proposed algorithm in small-scale networks where the number of users is 10. As shown in Fig.~\ref{conver}, the proposed algorithm can reach the convergence within 4 iterations. The result of the proposed algorithm can reach the near optimal value of exhaustive search algorithm. The gap between the average MOS of exhaustive search algorithm and proposed algorithm is less than $0.02$. Hence, \textbf{Remark~\ref{greed_remark}} and \textbf{Remark~\ref{tradeoff}} are both proved. We also compare the proposed algorithm with classic algorithm and show the improvement of algorithms of three sub-problems, i.e., one-to-one swap matching algorithm, greedy algorithm and Lagrange dual algorithm, respectively. From Fig.~\ref{conver}, the algorithms of three sub-problems can all improve the average MOS. The improvement brought by greedy algorithm is larger than another two sub-problem algorithms. The caching placement strategy makes the most important role in the proposed algorithm.\par
\begin{figure} [htbp!]
\centering
\includegraphics[width=0.6\linewidth]{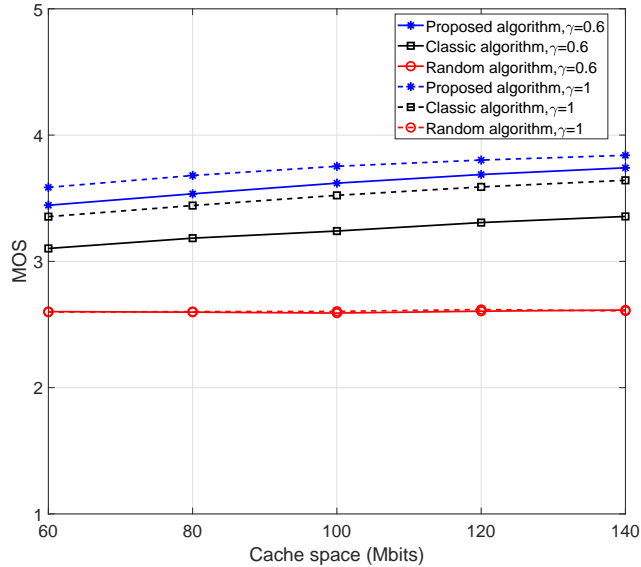}
 \caption{Average MOS of users with varying cache space.}\label{cs_mos}
 \vspace{-1.8em}
\end{figure}
\begin{figure} [htbp!]
\centering
\includegraphics[width=0.6\linewidth]{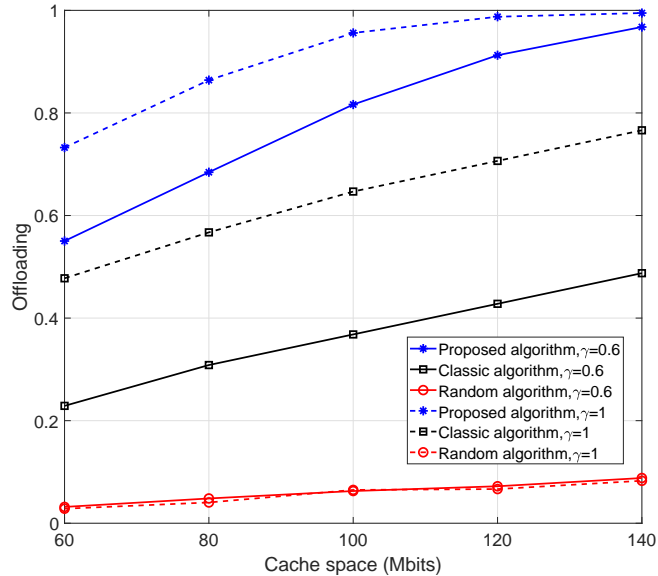}
 \caption {UAV backhaul traffic offloading ratio with varying cache space.}\label{cs_offloading}
 \vspace{-2em}
\end{figure}
Then we demonstrate the effectiveness of the proposed algorithm with varying cache space size. We set that the cache space of each UAV is ranging from 60 to 140~Mbits, the Zipf parameter $\gamma$ is 0.6 and 1, and the number of users is $100$.
The simulation results in Fig.~\ref{cs_mos} and Fig.~\ref{cs_offloading} show that the average MOS and UAV backhaul traffic offloading ratio of the proposed algorithm are improved compared with the classic algorithm and the random algorithm. The average MOS and UAV backhaul traffic offloading ratio increases as cache place increases. {For random algorithm, the average MOS almost remains static with varying cache space size and Zipf parameter, while the UAV backhaul traffic offloading ratio increases slowly as cache space increases.} Besides, the simulation results show that, compared with the cases of $\gamma=0.6$, all the three algorithms with $\gamma=1$ achieve better system performance. This is because the users have more requests concentrating on the most popular contents with $\gamma=1$ than $\gamma=0.6$, since the Zipf parameter $\gamma$ determines the skewness of content popularity. In Fig.~\ref{cs_offloading}, the traffic offloading of the proposed algorithm is close to $1$ when $\gamma=1,~H=140$~Mbits, which means that most of the contents requested by users have been cached in UAVs and do not need to be fetched from the MBS through the backhaul link of UAVs.\par
\begin{figure} [htbp!]
\centering
\includegraphics[width=0.6\linewidth]{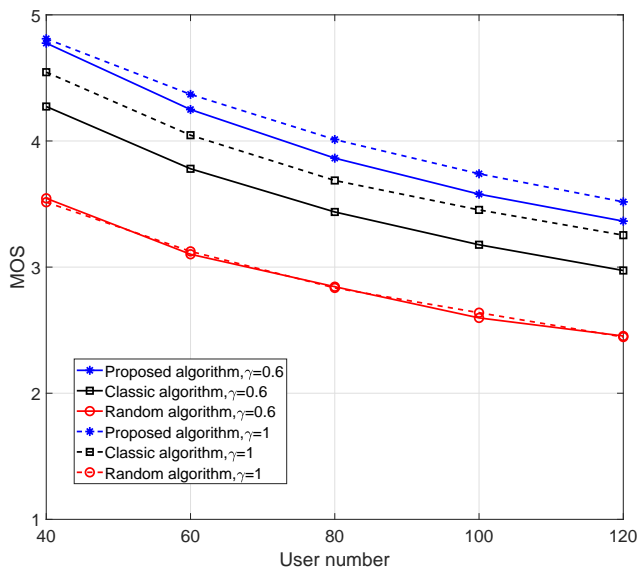}
 \caption{Average MOS of users with varying user number.}\label{ue_mos}
 \vspace{-1.8em}
\end{figure}
\begin{figure} [htbp!]
\centering
\includegraphics[width=0.6\linewidth]{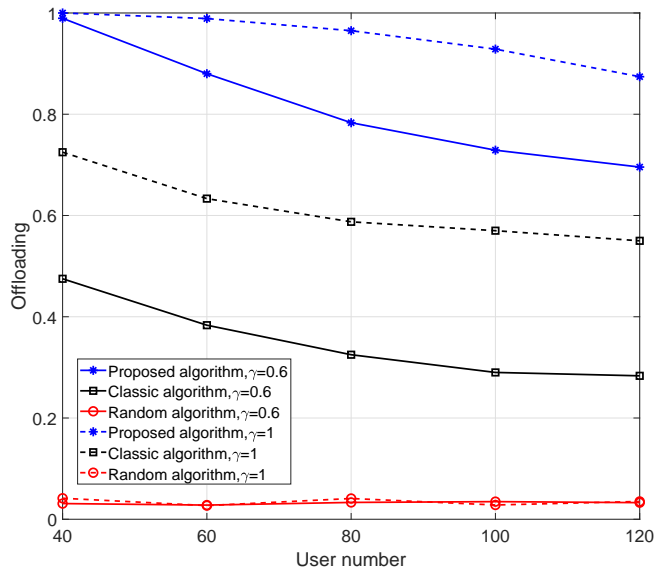}
 \caption{UAV backhaul traffic offloading ratio with varying user number.}\label{ue_offloading}
\vspace{-2em}
\end{figure}
\vspace{-0.1em}
Next, we demonstrate the effectiveness of the proposed algorithm with varying user number ranging from 40 to 120. In the simulation, we assume that the cache space of each UAV is $H = 100$~Mbits. Other parameters remain unchanged. As shown in Fig.~\ref{ue_mos} and Fig.~\ref{ue_offloading}, all the system performance indicators of the proposed algorithm are greatly improved compared with the other two benchmark algorithms. It is obvious that the system performance will be deteriorated as the number of users in the system increases. In Fig.~\ref{ue_mos}, the performance gap of the average MOS between $\gamma = 1$ and $\gamma = 0.6$ of the proposed algorithm is smaller than that of the classic algorithm. This result indicates the necessity and advantage of the proposed caching placement in our algorithm. When user number is $40$, the average MOS of the proposed algorithm have almost no difference with $\gamma = 0.6$ and $\gamma = 1$. For the random algorithm, the Zipf parameter $\gamma$ has almost no effect on the system performance indicators. In Fig.~\ref{ue_offloading}, the UAV backhaul traffic offloading ratio of the classic algorithm and the proposed algorithm with $\gamma = 1$ decreases more slowly than that with $\gamma = 0.6$ since the requests of users are more concentrated with $\gamma=1$, meanwhile, the UAV backhaul traffic offloading ratio of the random algorithm is almost unchanged as the number of users increases.\par
\begin{figure} [htbp!]
\centering
\includegraphics[width=0.6\linewidth]{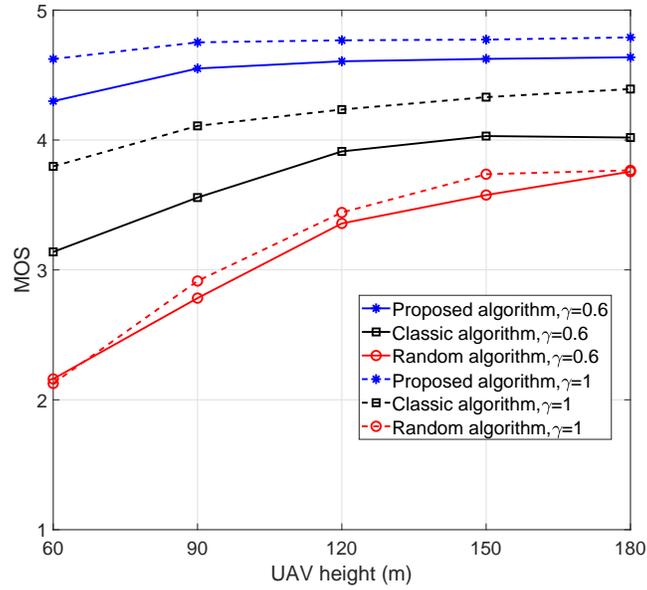}
 \caption{Average MOS of users with varying UAV height.}\label{height_mos}
 \vspace{-1.8em}
\end{figure}
\begin{figure} [htbp!]
\centering
\includegraphics[width=0.6\linewidth]{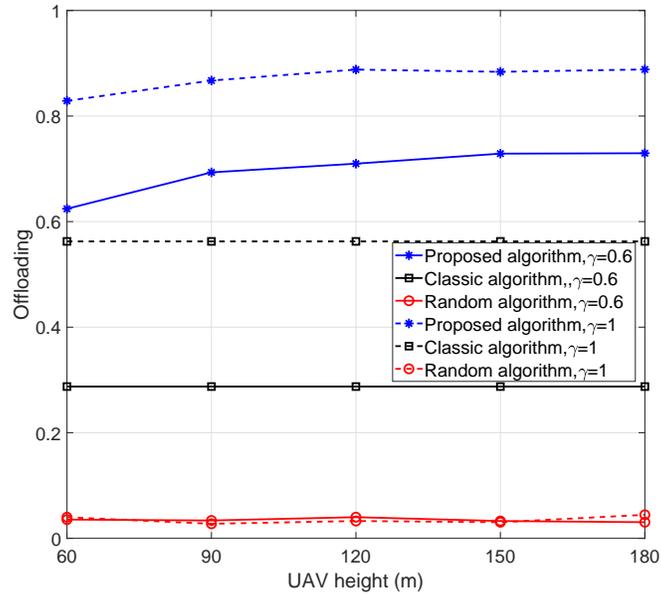}
 \caption{UAV backhaul traffic offloading ratio with varying UAV height.}\label{height_offloading}
\vspace{-2em}
\end{figure}
\vspace{-0.1em} We consider the impact of the UAV height on the performance of our proposed algorithm. UAV height varies from $60$ to $180$ and other parameters remain unchanged. As shown in Fig.~\ref{height_mos} and Fig.~\ref{height_offloading}, the system performance increases when UAV height is between $60$ and $120$. When UAV height is ranging from $120$ to $180$, the performance gains almost remain unchanged. \textcolor{blue}{According to~\eqref{possibility}, the probability of Los link between user and UAV increases with UAV height increases. The pathloss under LoS link is smaller than that of NloS link. So the system performance will be better under LoS link. According to~\eqref{pathloss}, it is obvious that the increase of UAV has little impact on the pathloss since $log_{10}h$ increases slowly as UAV height increases. Increasing the UAV height helps improve system performance at low-to-medium UAV altitude, but does not have a great effect on them in the high altitude.} Proper UAV height can significantly yield system performance gains.\par
\begin{figure} [htbp!]
\centering
\includegraphics[width=0.6\linewidth]{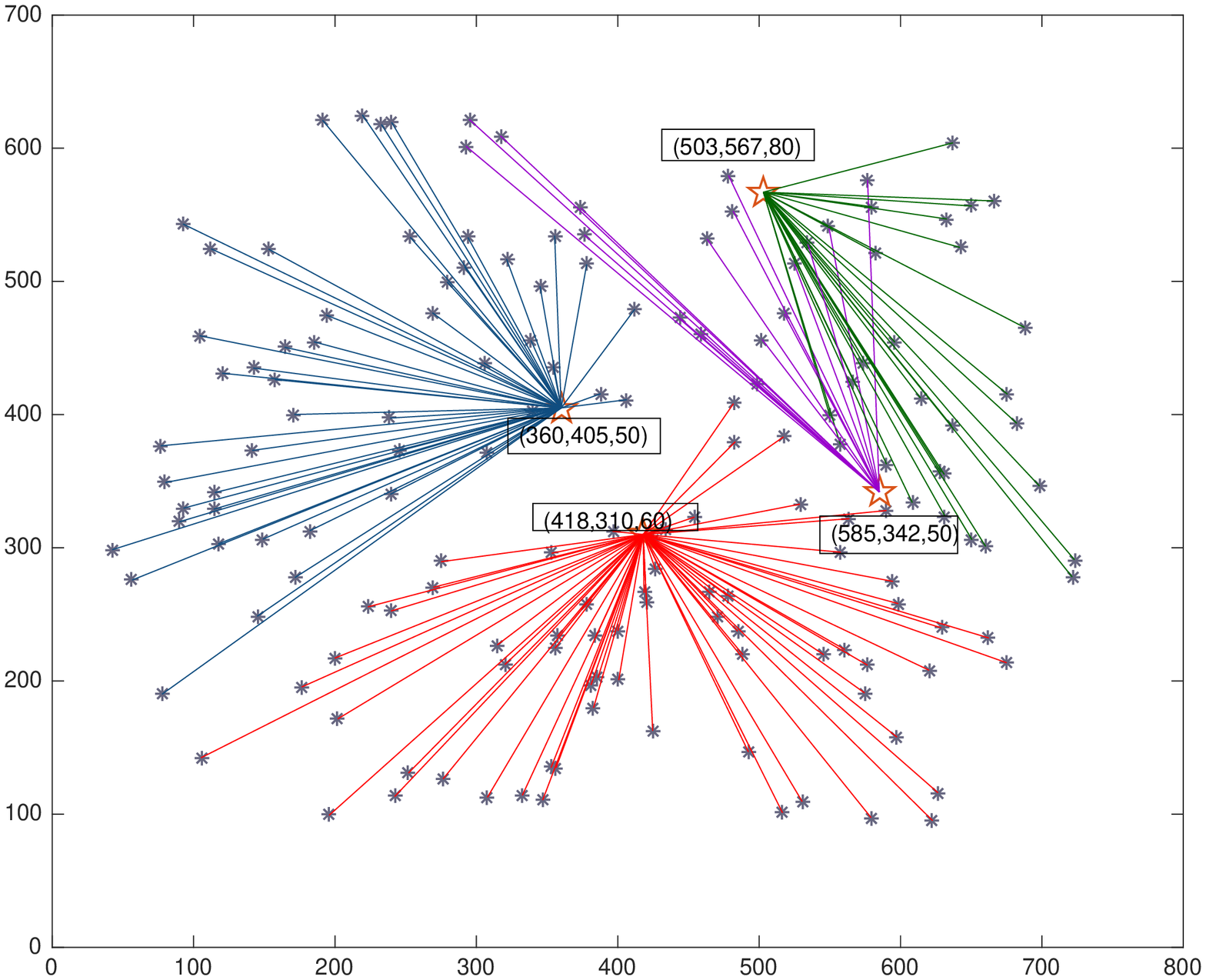}
 \caption{User association by Lagrange dual algorithm.}\label{la}
 \vspace{-1.8em}
\end{figure}
\begin{figure} [htbp!]
\centering
\includegraphics[width=0.6\linewidth]{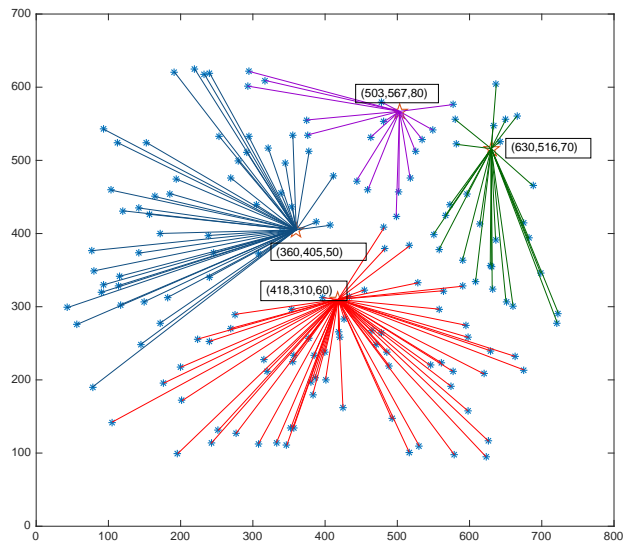}
 \caption{User association by max C/I algorithm.}\label{maxci}
\vspace{-2em}
\end{figure}
\vspace{-0.1em} The solution of user association by Lagrange dual algorithm and max C/I algorithm are shown in Fig.~\ref{la} and Fig.~\ref{maxci}, respectively. Fig.~\ref{conver} shows that user association can significantly improve system performance.  For convenience, the terrestrial location of UAV is shown. In Fig.~\ref{la}, the UAV deployment, cache placement and user association influence each other, which also shows the necessity of our research. However, users are usually associated with nearby UAV in Fig.~\ref{maxci} with max C/I Algorithm. \textcolor{blue}{The max C/I Algorithm, associated the nearest UAV, is suboptimal since it takes only SINR into account and ignores the bandwidth allocation for each user. According to~(2), the downlink transmission rate is related to both SINR and bandwidth allocation. When a large number of users are associated with one UAV, the downlink transmission rate of each user will be greatly reduced.} It also reveals that user association is quite important in the UAV-assisted cellular network design.
\section{Conclusion}
In this paper, we have investigated the joint optimization of UAV deployment, caching placement and user association in UAV-assisted cellular networks. {We formulated an optimization problem and proposed a low complexity suboptimal algorithm since the formulated problem is a combinatorial non-convex optimization problem. We demonstrated the convergence and network performance to verify the feasibility and effectiveness of the proposed algorithm by simulation results. From the simulation results, the caching placement can greatly improve the system performance, which also confirms the advantages of edge caching. The UAV placement and resource allocation are vital for providing an excellent channel condition to users when designing the UAV-assisted cellular network. In the future work, we would pay more attention to mobile UAV scenario. Multiple UAVs' trajectory optimization without collision and power control is a promising research direction to improve the coverage and QoE of users.}

\begin{appendices}
\section*{Appendix~A: Proof of Theorem~\ref{lemma_cp}}\label{greedy_prove}
\renewcommand{\theequation}{A.\arabic{equation}}
\setcounter{equation}{0}
\begin{proof}
Since non-negative linear combination of multiple monotone and submodular functions is closed, we only need to prove $Q_{m,k}\left( u\right )$ is a monotone and submodular function.\par
\begin{equation}
Q_{m,k} =  - \ln \left( {D_{m,k}^a + \frac{{\sum\nolimits_{f=1}^F {\left( {1 - {u_{m,f}}} \right)s{q_{k,f}}} }}{b_{m,n}}} \right).
\end{equation}
Since $\mathcal{X}$ and $\mathcal{A}$ are both fixed matrix, the downlink radio transmission delay $D_{m,k}^a$ is a constant according to~\eqref{Da}. For $\mathcal{X} \subseteq \mathcal{Y} \subseteq \mathcal{T}$, it is obvious ${Q_{m,k}}\left( \mathcal{X} \right) \le {Q_{m,k}}\left( \mathcal{Y} \right)$. So $Q_{m,k}\left( u\right )$ is a monotonically increasing function.\par
When ${\cal X}{\rm{ = }}{\cal Y} \subseteq {\cal T}$, the marginal gain obtained by adding $\left\{ x \right\}$ is the same.
\begin{equation}
{Q_{m,k}}\left( {{\cal X} \cup \left\{ x \right\}} \right) - {Q_{m,k}}\left( {\cal X} \right) \ge {Q_{m,k}}\left( {{\cal Y} \cup \left\{ x \right\}} \right) - {Q_{m,k}}\left( {\cal Y} \right).
\end{equation}
Thus the condition is satisfies.
\par
When ${\cal X} \subset {\cal Y} \subseteq {\cal T}$, the contents requested by users can affect the value of $Q_{m,k}$ and a user can only request a content. We discuss the effect of adding $\left\{ x \right\}$ on the marginal gain in classification.
\begin{enumerate}[(1)]
  \item If $\left\{ x \right\}$ contains the content requested by user $k$, namely the content is not contained in ${\cal X}$ and ${\cal Y}$.
  \begin{equation}
  {Q_{m,k}}\left( {\cal X} \right){\rm{ = }}{Q_{m,k}}\left( {\cal Y} \right){\rm{ = }} - \ln \left( {D_{m,k}^a + \frac{{\sum\nolimits_{f=1}^F {s{q_{k,f}}} }}{{{b_{m,n}}}}} \right),
  \end{equation}
  \begin{equation}
  {Q_{m,k}}\left( {{\cal X} \cup x} \right){\rm{ = }}{Q_{m,k}}\left( {{\cal Y} \cup x} \right){\rm{ = }} - \ln \left( {D_{m,k}^a} \right).
  \end{equation}
  Then, we have
  \begin{equation}
  {Q_{m,k}}\left( {{\cal X} \cup \left\{ x \right\}} \right) - {Q_{m,k}}\left( {\cal X} \right) \ge {Q_{m,k}}\left( {{\cal Y} \cup \left\{ x \right\}} \right) - {Q_{m,k}}\left( {\cal Y} \right).
  \end{equation}
  \item We assume that the content requested by user $k$ is not contained in $\left\{ x \right\}$, ${\cal X}$, ${\cal Y}$.
  Then
  \begin{equation}
  \begin{aligned}
  {Q_{m,k}}\left( {\cal X} \right)&{\rm{ = }}{Q_{m,k}}\left( {\cal Y} \right){\rm{ = }}{Q_{m,k}}\left( {{\cal X} \cup x} \right){\rm{ = }}{Q_{m,k}}\left( {{\cal Y} \cup x} \right)\\
  &{\rm{ = }} - \ln \left( {D_{m,k}^a + \frac{{\sum\nolimits_{f=1}^F {s{q_{k,f}}} }}{{{b_{m,n}}}}} \right).
  \end{aligned}
  \end{equation}
  Then, we have
  \begin{equation}
  {Q_{m,k}}\left( {{\cal X} \cup \left\{ x \right\}} \right) - {Q_{m,k}}\left( {\cal X} \right) \ge {Q_{m,k}}\left( {{\cal Y} \cup \left\{ x \right\}} \right) - {Q_{m,k}}\left( {\cal Y} \right).
  \end{equation}
  \item The content requested by user $k$ is not contained in $\left\{ x \right\}$. The content is contained in ${\cal X}$ and ${\cal Y}$. Then
  \begin{equation}
  \begin{aligned}
  {Q_{m,k}}\left( {\cal X} \right)&{\rm{ = }}{Q_{m,k}}\left( {\cal Y} \right){\rm{ = }}{Q_{m,k}}\left( {{\cal X} \cup x} \right){\rm{ = }}{Q_{m,k}}\left( {{\cal Y} \cup x} \right)\\
  &{\rm{ = }} - \ln \left( {D_{m,k}^a}\right).
  \end{aligned}
  \end{equation}
   Then, we have:
  \begin{equation}
  {Q_{m,k}}\left( {{\cal X} \cup \left\{ x \right\}} \right) - {Q_{m,k}}\left( {\cal X} \right) \ge {Q_{m,k}}\left( {{\cal Y} \cup \left\{ x \right\}} \right) - {Q_{m,k}}\left( {\cal Y} \right).
  \end{equation}
  \item The content requested by user $k$ is not contained in $\left\{ x \right\}$ and ${\cal X}$. It is contained in ${\cal Y}$.
  \begin{equation}
  \begin{aligned}
{Q_{m,k}}\left( {\cal X} \right)&{\rm{ = }}{Q_{m,k}}\left( {{\cal X} \cup x} \right)\\
&{\rm{ = }} - \ln \left( {D_{m,k}^a + \frac{{\sum\nolimits_{f=1}^F {s{q_{k,f}}} }}{{{b_{m,n}}}}} \right).
\end{aligned}
\end{equation}
  \begin{equation}
  {Q_{m,k}}\left( {\cal Y} \right){\rm{ = }}{Q_{m,k}}\left( {{\cal Y} \cup x} \right){\rm{ = }} - \ln \left( {D_{m,k}^a} \right).
  \end{equation}
  Then, we have:
  \begin{equation}
  {Q_{m,k}}\left( {{\cal X} \cup \left\{ x \right\}} \right) - {Q_{m,k}}\left( {\cal X} \right) \ge {Q_{m,k}}\left( {{\cal Y} \cup \left\{ x \right\}} \right) - {Q_{m,k}}\left( {\cal Y} \right).
  \end{equation}
 \end{enumerate}
 In summary, ${Q_{m,k}}\left( u \right)$ is a monotone and submodular function. So $ {\sum\limits_{k = 1}^K {{a_{m,k}}{Q_{m,k}}} } $ is also a monotone and submodular function. The optimization objective of the caching placement is to maximize a monotone and submodular function.\par
 It is shown in~\cite{Calinescu2007Maximizing} that the greedy algorithm for maximizing a monotone and submodular function can reach a near-optimal solution. Hence, theorem~\ref{lemma_cp} is proved.
\end{proof}
\section*{Appendix~B: Proof of Theorem~\ref{theorem}}\label{appendix_b}
\renewcommand{\theequation}{B.\arabic{equation}}
\setcounter{equation}{0}
\begin{proof}
The derivative of $L\left( \alpha  \right)$ is given by
\begin{equation}
\label{daoshu}
\frac{{\partial L}}{{\partial {\alpha _m}}}\left( \alpha  \right){\rm{ = }}{w_m}\left( \alpha \right) - \sum\limits_{k = 1}^K {{a_{m,k}}\left( \alpha \right)}.
\end{equation}
In our optimization problem~\eqref{optimal_user_association}, we have ${w_m} = \sum\limits_k {{a_{m,k}}}$. According to~\eqref{daoshu}, when $w_m$ and $\sum\limits_k {{a_{m,k}}}$ are bounded, the subgradient of the dual objective function $\partial L$ is also bounded
\begin{equation}
\mathop {\sup }\limits_t \left\{ {\left\| {\partial L\left( {\alpha\left( t \right)} \right)} \right\|} \right\} \le a,
\end{equation}
where $a$ is some scalar. The optimization problem satisfied the necessary conditions of Proposition 6.3.6 in~\cite{Pardalos2010Convex}. Theorem~\ref{theorem} is proved by applying this proposition.\par
\end{proof}

\section*{Appendix~C: Proof of Theorem~\ref{theorem3}\label{appendix_c}}
\renewcommand{\theequation}{C.\arabic{equation}}
\setcounter{equation}{0}
\begin{proof} {
We assume Algorithm~\ref{cache_placement} yields a decreasing objective value in $i-th$ iteration, which is denoted as
\begin{equation}
M\left( {{{\cal X}^{\left( i \right)}},{{\cal S}^{\left( i \right)}},{{\cal A}^{\left( {i - 1} \right)}}} \right) < M\left( {{{\cal X}^{\left( i \right)}},{{\cal S}^{\left( {i - 1} \right)}},{{\cal A}^{\left( {i - 1} \right)}}} \right),
\end{equation}
This means that the benefits of at least one content are reduced. In our system model, the benefit of content $f$ cached in UAV $m$ can be calculated as a constant with fixed $\cal X$ and $\cal A$, which is denoted as
\begin{equation}
\Omega _{m,f} = \sum\limits_{k = 1}^K {{a_{m,k}}{q_{k,f}}\ln \left( {\frac{{B{{\log }_2}\left( {1 + SIN{R_{m,n,k}}} \right)}}{{s\sum\limits_{k = 1}^K {{a_{m,k}}} }}} \right)},
\end{equation}
As we have discussed above, one-to-one swap matching algorithm and Lagrange dual algorithm can both yield an increasing objective value in each iteration. Then
\begin{equation}
{\Omega _{m,f}}\left( {i - 1} \right) \le {\Omega _{m,f}}\left( i \right),\forall m,f,
\end{equation}
which contradicts our assumption. So Theorem~\ref{theorem3} is proved by contradiction.}
\end{proof}
\section*{Appendix~D: Proof of Convergence\label{appendix_d}}
\renewcommand{\theequation}{D.\arabic{equation}}
\setcounter{equation}{0}
\begin{proof} {
Let $M\left( {{{\cal X}^{\left( i \right)}},{{\cal S}^{\left( i \right)}},{{\cal A}^{\left( i \right)}}} \right)$ denote the total MOS of all users calculated by ${{{\cal X}^{\left( i \right)}}}$, ${{{\cal S}^{\left( i \right)}}}$ and ${{{\cal A}^{\left( i \right)}}}$. We obtain ${{{\cal X}^{\left( i \right)}}}$ under a stable matching state by \textbf{Algorithm 1} with fixed ${{{\cal S}^{\left( i - 1 \right)}}}$ and ${{{\cal A}^{\left( i - 1 \right)}}}$,
\begin{equation}
M\left( {{{\cal X}^{\left( i \right)}},{{\cal S}^{\left( i - 1 \right)}},{{\cal A}^{\left( i - 1  \right)}}} \right) \ge M\left( {{{\cal X}^{\left( i \right)}},{{\cal S}^{\left( {i - 1} \right)}},{{\cal A}^{\left( {i - 1} \right)}}} \right).
\end{equation}
In \textbf{Algorithm 4}, based on Theorem~3, the total MOS of all users does not decrease after \textbf{Algorithm~2} is performed.
\begin{equation}
M\left( {{{\cal X}^{\left( i \right)}},{{\cal S}^{\left( i \right)}},{{\cal A}^{\left( {i - 1} \right)}}} \right) \ge M\left( {{{\cal X}^{\left( i \right)}},{{\cal S}^{\left( {i - 1} \right)}},{{\cal A}^{\left( {i - 1} \right)}}} \right).
\end{equation}
\textbf{Algorithm 4} is guaranteed to converge to a suboptimal solution by Lagrange dual algorithm with fixed ${{{\cal X}^{\left( i - 1 \right)}}}$ and ${{{\cal S}^{\left( i - 1 \right)}}}$ .
\begin{equation}
M\left( {{{\cal X}^{\left( i \right)}},{{\cal S}^{\left( i \right)}},{{\cal A}^{\left( i \right)}}} \right) \ge M\left( {{{\cal X}^{\left( i \right)}},{{\cal S}^{\left( i \right)}},{{\cal A}^{\left( {i - 1} \right)}}} \right).
\end{equation}
Moreover, the average MOS of users is bounded in a practical UAV-assisted cellular system. So the proposed algorithm can reach convergence and obtain a suboptimal solution.}
\end{proof}
\end{appendices}{}

\vspace{-0.5em}
\begin{spacing}{1}
\bibliography{mybib}{}

\begin{thebibliography}{10}
\providecommand{\url}[1]{#1}
\csname url@samestyle\endcsname
\providecommand{\newblock}{\relax}
\providecommand{\bibinfo}[2]{#2}
\providecommand{\BIBentrySTDinterwordspacing}{\spaceskip=0pt\relax}
\providecommand{\BIBentryALTinterwordstretchfactor}{4}
\providecommand{\BIBentryALTinterwordspacing}{\spaceskip=\fontdimen2\font plus
\BIBentryALTinterwordstretchfactor\fontdimen3\font minus
  \fontdimen4\font\relax}
\providecommand{\BIBforeignlanguage}[2]{{%
\expandafter\ifx\csname l@#1\endcsname\relax
\typeout{** WARNING: IEEEtran.bst: No hyphenation pattern has been}%
\typeout{** loaded for the language `#1'. Using the pattern for}%
\typeout{** the default language instead.}%
\else
\language=\csname l@#1\endcsname
\fi
#2}}
\providecommand{\BIBdecl}{\relax}
\BIBdecl

\bibitem{6477825}
S.~{Morgenthaler}, T.~{Braun}, Z.~{Zhao}, T.~{Staub}, and M.~{Anwander},
  ``{UAVN}et: A mobile wireless mesh network using unmanned aerial vehicles,''
  in \emph{IEEE Global Commun. Conf. (GLOBECOM) Workshop}, Dec. 2012, pp.
  1603--1608.

\bibitem{7317490}
L.~{Gupta}, R.~{Jain}, and G.~{Vaszkun}, ``Survey of important issues in {UAV}
  communication networks,'' \emph{IEEE Commun. Surveys Tuts.}, vol.~18, no.~2,
  pp. 1123--1152, May 2016.

\bibitem{6978873}
S.~{Kandeepan}, K.~{Gomez}, L.~{Reynaud}, and T.~{Rasheed},
  ``Aerial-terrestrial communications: terrestrial cooperation and
  energy-efficient transmissions to aerial base stations,'' \emph{IEEE Trans.
  Aerosp. Electron. Syst.}, vol.~50, no.~4, pp. 2715--2735, Oct. 2014.

\bibitem{whitepaper}
``Cisco visual networking index: Global mobile data traffic forecast update,
  2016 to 2021 white paper,''
  \url{https://www.cisco.com/c/en/us/solutions/collateral/service-provider/visual-networking-index-vni/mobile-white-paper-c11-520862.html},
  Mar. 2017.

\bibitem{6871674}
E.~{Bastug}, M.~{Bennis}, and M.~{Debbah}, ``Living on the edge: The role of
  proactive caching in {5G} wireless networks,'' \emph{IEEE Commun. Mag.},
  vol.~52, no.~8, pp. 82--89, Aug. 2014.

\bibitem{8945154}
T.~{Zhang}, X.~{Fang}, Y.~{Liu}, and A.~{Nallanathan}, ``Content-centric mobile
  edge caching,'' \emph{IEEE Access}, pp. 11\,722--11\,731, 2019.

\bibitem{8868205}
T.~{Zhang}, X.~{Fang}, Y.~{Liu}, G.~Y. {Li}, and W.~{Xu}, ``{D2D}-enabled
  mobile user edge caching: A multi-winner auction approach,'' \emph{IEEE
  Trans. Veh. Technol.}, vol.~68, no.~12, pp. 12\,314--12\,328, Dec. 2019.

\bibitem{8614433}
M.~{Chen}, W.~{Saad}, and C.~{Yin}, ``Liquid state machine learning for
  resource and cache management in {LTE-U} unmanned aerial vehicle ({UAV})
  networks,'' \emph{IEEE Trans. Wireless Commun.}, vol.~18, no.~3, pp.
  1504--1517, Mar. 2019.

\bibitem{8576651}
B.~{Jiang}, J.~{Yang}, H.~{Xu}, H.~{Song}, and G.~{Zheng}, ``Multimedia data
  throughput maximization in {I}nternet-of-{T}hings system based on
  optimization of cache-enabled {UAV},'' \emph{IEEE Internet Things J.},
  vol.~6, no.~2, pp. 3525--3532, Apr. 2019.

\bibitem{8717714}
M.~{Chen}, W.~{Saad}, and C.~{Yin}, ``Echo-liquid state deep learning for 360°
  content transmission and caching in wireless {VR} networks with
  cellular-connected {UAV}s,'' \emph{IEEE Trans. Commun.}, vol.~67, no.~9, pp.
  6386--6400, Sep. 2019.

\bibitem{8603721}
N.~{Zhao}, F.~R. {Yu}, L.~{Fan}, Y.~{Chen}, J.~{Tang}, A.~{Nallanathan}, and
  V.~C.~M. {Leung}, ``Caching unmanned aerial vehicle-enabled small-cell
  networks: Employing energy-efficient methods that store and retrieve popular
  content,'' \emph{IEEE Veh. Technol. Mag.}, vol.~14, no.~1, pp. 71--79, Mar.
  2019.

\bibitem{article}
T.~Fang, H.~Tian, X.~Zhang, X.~Chen, X.~Shao, and Y.~Zhang, ``Context-aware
  caching distribution and {UAV} deployment: A game-theoretic approach,''
  \emph{Appl. Sci.}, vol.~8, p. 1959, Oct. 2018.

\bibitem{8254370}
N.~{Zhao}, F.~{Cheng}, F.~R. {Yu}, J.~{Tang}, Y.~{Chen}, G.~{Gui}, and
  H.~{Sari}, ``Caching {UAV} assisted secure transmission in hyper-dense
  networks based on interference alignment,'' \emph{IEEE Trans. Commun.},
  vol.~66, no.~5, pp. 2281--2294, May 2018.

\bibitem{8626132}
F.~{Cheng}, G.~{Gui}, N.~{Zhao}, Y.~{Chen}, J.~{Tang}, and H.~{Sari},
  ``{UAV}-relaying-assisted secure transmission with caching,'' \emph{IEEE
  Trans. Commun.}, vol.~67, no.~5, pp. 3140--3153, May 2019.

\bibitem{5430138}
P.~{Brooks} and B.~{Hestnes}, ``User measures of quality of experience: why
  being objective and quantitative is important,'' \emph{IEEE Netw.}, vol.~24,
  no.~2, pp. 8--13, Mar. 2010.

\bibitem{8867956}
X.~{Zhong}, Y.~{Guo}, N.~{Li}, Y.~{Chen}, and S.~{Li}, ``Deployment
  optimization of {UAV} relay for malfunctioning base station: Model-free
  approaches,'' \emph{IEEE Trans. Veh. Technol.}, vol.~68, no.~12, pp.
  11\,971--11\,984, Dec. 2019.

\bibitem{8776639}
Y.~{Huang}, W.~{Mei}, J.~{Xu}, L.~{Qiu}, and R.~{Zhang}, ``Cognitive {UAV}
  communication via joint maneuver and power control,'' \emph{IEEE Trans.
  Commun.}, vol.~67, no.~11, pp. 7872--7888, Nov. 2019.

\bibitem{7486987}
M.~{Mozaffari}, W.~{Saad}, M.~{Bennis}, and M.~{Debbah}, ``Efficient deployment
  of multiple unmanned aerial vehicles for optimal wireless coverage,''
  \emph{IEEE Commun. Lett.}, vol.~20, no.~8, pp. 1647--1650, Aug. 2016.

\bibitem{9003500}
X.~{Zhong}, Y.~{Guo}, N.~{Li}, and Y.~{Chen}, ``Joint optimization of relay
  deployment, channel allocation, and relay assignment for {UAV}s-aided {D2D}
  networks,'' \emph{IEEE/ACM Trans. Networking}, vol.~28, no.~2, pp. 804--817,
  Apr. 2020.

\bibitem{8247211}
Q.~{Wu}, Y.~{Zeng}, and R.~{Zhang}, ``Joint trajectory and communication design
  for multi-{UAV} enabled wireless networks,'' \emph{IEEE Trans. Wireless
  Commun.}, vol.~17, no.~3, pp. 2109--2121, Mar. 2018.

\bibitem{8432499}
Y.~{Cai}, F.~{Cui}, Q.~{Shi}, M.~{Zhao}, and G.~Y. {Li}, ``Dual-{UAV}-enabled
  secure communications: Joint trajectory design and user scheduling,''
  \emph{IEEE J. Sel. Areas Commun.}, vol.~36, no.~9, pp. 1972--1985, Sep. 2018.

\bibitem{9013181}
H.~{Shiri}, J.~{Park}, and M.~{Bennis}, ``Massive autonomous {UAV} path
  planning: A neural network based mean-field game theoretic approach,'' in
  \emph{2019 IEEE Global Commun. Conf. (GLOBECOM)}, 2019, pp. 1--6.

\bibitem{7738405}
J.~{Li} and Y.~{Han}, ``Optimal resource allocation for packet delay
  minimization in multi-layer {UAV} networks,'' \emph{IEEE Commun. Lett.},
  vol.~21, no.~3, pp. 580--583, March 2017.

\bibitem{8723343}
A.~{Gao}, Y.~{Hu}, W.~{Liang}, Y.~{Lin}, L.~{Li}, and X.~{Li}, ``A
  {QoE}-oriented scheduling scheme for energy-efficient computation offloading
  in {UAV} cloud system,'' \emph{IEEE Access}, vol.~7, pp. 68\,656--68\,668,
  2019.

\bibitem{8648498}
Y.~{Sun}, D.~{Xu}, D.~W.~K. {Ng}, L.~{Dai}, and R.~{Schober}, ``Optimal
  3{D}-trajectory design and resource allocation for solar-powered {UAV}
  communication systems,'' \emph{IEEE Trans. Commun.}, vol.~67, no.~6, pp.
  4281--4298, Jun. 2019.

\bibitem{8629316}
N.~{Zhao}, X.~{Pang}, Z.~{Li}, Y.~{Chen}, F.~{Li}, Z.~{Ding}, and M.~{Alouini},
  ``Joint trajectory and precoding optimization for {UAV}-assisted {NOMA}
  networks,'' \emph{IEEE Trans. Commun.}, vol.~67, no.~5, pp. 3723--3735, May
  2019.

\bibitem{8417640}
H.~{Ghazzai}, A.~{Kadri}, M.~{Ben Ghorbel}, and H.~{Menouar}, ``Optimal
  sequential and parallel {UAV} scheduling for multi-event applications,'' in
  \emph{Proc. 2018 IEEE 87th Veh. Technology Conf. (VTC Spring)}, Jun. 2018,
  pp. 1--6.

\bibitem{7438747}
Z.~{Zhao}, M.~{Peng}, Z.~{Ding}, W.~{Wang}, and H.~V. {Poor}, ``Cluster content
  caching: An energy-efficient approach to improve quality of service in cloud
  radio access networks,'' \emph{IEEE J. Sel. Areas Commun.}, vol.~34, no.~5,
  pp. 1207--1221, May 2016.

\bibitem{3gpp}
3GPP, ``Study on enhanced {LTE} support for aerial vehicles,'' \emph{TR 36.777,
  Release 15}, pp. 33--46, Dec. 2017.

\bibitem{9018112}
H.~{Wang}, J.~{Wang}, G.~{Ding}, J.~{Chen}, and J.~{Yang}, ``Completion time
  minimization for turning angle-constrained {UAV-to-UAV} communications,''
  \emph{IEEE Trans. Veh. Technol.}, vol.~69, no.~4, pp. 4569--4574, Apr. 2020.

\bibitem{8038869}
M.~{Mozaffari}, W.~{Saad}, M.~{Bennis}, and M.~{Debbah}, ``Mobile unmanned
  aerial vehicles ({UAV}s) for energy-efficient {Internet of Things}
  communications,'' \emph{IEEE Trans. Wireless Commun.}, vol.~16, no.~11, pp.
  7574--7589, Nov. 2017.

\bibitem{6497017}
Q.~{Ye}, B.~{Rong}, Y.~{Chen}, M.~{Al-Shalash}, C.~{Caramanis}, and J.~G.
  {Andrews}, ``User association for load balancing in heterogeneous cellular
  networks,'' \emph{IEEE Trans. Wireless Commun.}, vol.~12, no.~6, pp.
  2706--2716, Jun. 2013.

\bibitem{Breslau2002Web}
L.~Breslau, P.~Cao, L.~Fan, G.~Phillips, and S.~Shenker, ``Web caching and
  zipf-like distributions: evidence and implications,'' in \emph{Proc. IEEE
  INFOCOM}, Mar. 1999, pp. 126--134.

\bibitem{6175019}
M.~{Taghizadeh}, K.~{Micinski}, S.~{Biswas}, C.~{Ofria}, and E.~{Torng},
  ``Distributed cooperative caching in social wireless networks,'' \emph{IEEE
  Trans. Mobile Comput.}, vol.~12, no.~6, pp. 1037--1053, Jun. 2013.

\bibitem{6877621}
M.~{Rugelj}, U.~{Sedlar}, M.~{Volk}, J.~{Sterle}, M.~{Hajdinjak}, and A.~{Kos},
  ``Novel cross-layer {QoE}-aware radio resource allocation algorithms in
  multiuser {OFDMA} systems,'' \emph{IEEE Trans. Commun.}, vol.~62, no.~9, pp.
  3196--3208, Sep. 2014.

\bibitem{booknp}
M.~R.~Garey and D.~Johnson, \emph{Computers and Intracdtability: A Guide to the
  Theory of NP-Completeness}.\hskip 1em plus 0.5em minus 0.4em\relax W. H.
  Freeman, 1979.

\bibitem{inproceedings}
E.~Bodine-Baron, C.~Lee, A.~Chong, B.~Hassibi, and A.~Wierman, ``Peer effects
  and stability in matching markets,'' in \emph{Algorithmic Game Theory}, 2011,
  pp. 117--129.

\bibitem{8415760}
J.~{Cui}, Y.~{Liu}, Z.~{Ding}, P.~{Fan}, and A.~{Nallanathan}, ``{QoE}-based
  resource allocation for multi-cell {NOMA} networks,'' \emph{IEEE Trans.
  Wireless Commun.}, vol.~17, no.~9, pp. 6160--6176, Sep. 2018.

\bibitem{8422358}
J.~{Zhao}, Y.~{Liu}, T.~{Mahmoodi}, K.~K. {Chai}, Y.~{Chen}, and Z.~{Han},
  ``Resource allocation in cache-enabled {CRAN} with non-orthogonal multiple
  access,'' in \emph{IEEE Proc. of Int. Commun. Conf. (ICC)}, May 2018, pp.
  1--6.

\bibitem{stable}
D.~Gale and L.~S.~Shapley, ``College admissions and stability of marriage,''
  \emph{American Mathematical Monthly}, vol.~69, pp. 9--15, May 2013.

\bibitem{article1}
S.~Boyd and A.~Mutapcic, ``Subgradient methods,'' \emph{lecture notes of
  EE392o, Stanford University, Autumn Quarter}, vol. 2004, Jan. 2003.

\bibitem{Calinescu2007Maximizing}
G.~Calinescu, C.~Chekuri, M.~Pál, and J.~Vondrák, ``Maximizing a submodular
  set function subject to a matroid constraint,'' \emph{SIAM J. Comput.},
  vol.~40, no.~6, pp. 1740--1766, 2007.

\bibitem{Pardalos2010Convex}
D.~Bertsekas, \emph{Convex optimization theory}.\hskip 1em plus 0.5em minus
  0.4em\relax Athena Scientific, 2009.

\end{thebibliography}
\bibliographystyle{IEEEtran}
\end{spacing}
\end{document}